\newif\ifmaintext
\maintexttrue

\newif\ifsupplemental
\supplementaltrue

\documentclass[aps,prl,reprint,twocolumn,showpacs,superscriptaddress,groupedaddress]{revtex4}  

\newif\ifbothtexts
\ifmaintext \ifsupplemental
  \bothtextstrue
\else
  \bothtextsfalse
\fi \fi

\usepackage{graphicx}  
\usepackage{dcolumn}   
\usepackage{bm}        
\usepackage{amssymb,amsmath,amsthm}   
\usepackage{thmtools,thm-restate} 

\usepackage{tikz}
\usetikzlibrary{shapes,snakes,calc}

\usepackage{xcolor}

\usepackage{nicefrac}
\newcommand{\nfrac}[2]{\nicefrac{#1}{#2}}

\usepackage[inline]{enumitem}

\usepackage[normalem]{ulem}

\ifsupplemental
\ifmaintext
\else
\usepackage{xr}
\usepackage{xcite}
\externalcitedocument{contextual_fraction_submission_both}
\fi
\fi
\usepackage{hyperref}

\allowdisplaybreaks

\theoremstyle{plain}
\newtheorem{theorem}                 {Theorem}
\newtheorem{proposition}             {Proposition}
\newtheorem{lemma}      [proposition]{Lemma}

\theoremstyle{remark}
\newtheorem*{remark}{Remark}

\newcommand{\tuple}[1]{\mathopen{\langle}#1\mathclose{\rangle}} 
\usepackage{colonequals}
\newcommand{\defeq}{\colonequals} 

\newcommand{\setdef}  [2]{\left\{#1 \mid #2\right\}}             
\newcommand{\enset}   [1]{\mathopen{ \{ }#1\mathclose{ \} }} 
\newcommand{\fdec}    [3]{#1: #2 \longrightarrow #3}
\newcommand{\fdef}    [3]{#1:: #2 \longmapsto     #3}

\newcommand{\blank}{{-}}

\newcommand{\ZZ}  {\mathbb{Z}}

\newcommand{\RR}  {\mathbb{R}}

\newcommand{\RRpz}{\mathbb{R}_{\geq 0}} 

\newcommand{\Forall}[1]{\forall {#1},\;} 

\newenvironment{calculation}{\begin{eqnarray*}&&}{\end{eqnarray*}}
\newenvironment{calculationNUM}{\begin{eqnarray}&&}{\end{eqnarray}}
\newcommand{\just}[2]{\\ &#1& \rule{2em}{0pt} \{ \mbox{\rule[-.7em]{0pt}{1.8em} \footnotesize #2 \/} \} \nonumber\\ && }
\newcommand{\ejust}[1]{\\ &#1& \nonumber\\ && }

\def\bimplies{\Leftrightarrow}

\newcommand{\Mcomma}{\text{ ,}}

\newcommand{\Mdot}{\text{ .}}
\newcommand{\Mand}{\quad\quad\text{ and }\quad\quad}

\newcommand{\M}{\mathcal{M}}
\newcommand{\XMO}{\tuple{X, \M, O}}
\newcommand{\XMOp}{\tuple{X, \M, O'}}
\newcommand{\XpMpO}{\tuple{X', \M', O}}
\newcommand{\XMOone}{\tuple{X_1, \M_1, O}}
\newcommand{\XMOtwo}{\tuple{X_2, \M_2, O}}
\newcommand{\XMOi}{\tuple{X_i, \M_i, O}}

\newcommand{\CF} {\mathsf{CF}}
\newcommand{\NCF}{\mathsf{NCF}}

\newcommand{\vect}[1]{\mathbf{#1}}

\newcommand{\av}{\vect{a}}
\newcommand{\bb}{\vect{b}}
\newcommand{\cc}{\vect{c}}
\newcommand{\dd}{\vect{d}}

\newcommand{\yy}{\vect{y}}

\newcommand{\ii}{\vect{i}}
\newcommand{\oo}{\vect{o}}
\newcommand{\qq}{\vect{q}}
\newcommand{\sv}{\vect{s}}
\newcommand{\vone} {\vect{1}}
\newcommand{\vzero}{\vect{0}}

\newcommand{\IM}{\vect{M}}
\newcommand{\ve}{\vect{v}^{e}}

\newcommand{\veNC}{\vect{v}^{e^{NC}}}
\newcommand{\veSC}{\vect{v}^{e^{SC}}}
\newcommand{\veA}[1]{\vect{v}^{#1}}

\newcommand{\MQ}{\vect{Q}}
\newcommand{\MT}{\vect{T}}
\newcommand{\MZ}{\vect{Z}}
\newcommand{\MIn}{\vect{I}_n}

\newcommand{\diag}{\mathsf{diag}}

\newcommand*\Distr{\mathcal{D}}

\newcommand{\weight}[1]{w(#1)}
\newcommand{\weightz}{w}

\newcommand{\textLP}{\texttt}
\newcommand{\transpose}[1]{#1^{\intercal}}

\newcommand{\BellState}{\ket{\Phi^+}}
\newcommand{\BellStateExpand}{\frac{\ket{0 0} + \ket{1 1}}{\sqrt{2}}}
\newcommand{\BellStateExpandLine}{\frac{1}{\sqrt{2}}(\ket{0 0} + \ket{1 1})}

\newcommand{\GHZStateN}[1]{\ket{\psi_{\mathsf{GHZ}(#1)}}}
\newcommand{\GHZStateNN}{\ket{\psi_{\mathsf{GHZ}(n)}}}

\newcommand{\GHZStateNExpand}{\frac{\ket{0}^{\otimes n}   +  \ket{1}^{\otimes n}}{\sqrt{2}}}

\newcommand{\ket}[1]{|#1\rangle}

\newcommand{\ltMBQC}{$l2$-MBQC}

\usepackage{stmaryrd}
\newcommand{\ds}[1]{\llbracket #1 \rrbracket}

\newcommand{\Oset}[1]{\mathcal{O}_{#1}}
\newcommand{\Lset}[1]{\mathcal{L}_{#1}}
\newcommand{\Rset}[1]{\mathcal{R}_{#1}}

\newcommand{\choice}{\mathbin{\&}}

\newcommand{\inv}{^{\raisebox{.2ex}{$\scriptscriptstyle-1$}}} 

\newcommand{\prlsection}[1]{\paragraph*{#1.---}\hspace{-1em}}

\def\emph{\textbf}
\def\stress{\textit}

\newcommand\pgfmathsinandcos[3]{%
  \pgfmathsetmacro#1{sin(#3)}%
  \pgfmathsetmacro#2{cos(#3)}%
}
\newcommand\LongitudePlane[3][current plane]{%
  \pgfmathsinandcos\sinEl\cosEl{#2} 
  \pgfmathsinandcos\sint\cost{#3} 
  \tikzset{#1/.estyle={cm={\cost,\sint*\sinEl,0,\cosEl,(0,0)}}}
}
\newcommand\LatitudePlane[3][current plane]{%
  \pgfmathsinandcos\sinEl\cosEl{#2} 
  \pgfmathsinandcos\sint\cost{#3} 
  \pgfmathsetmacro\yshift{\cosEl*\sint}
  \tikzset{#1/.estyle={cm={\cost,0,0,\cost*\sinEl,(0,\yshift)}}} %
}

\newcommand\DrawLatitudeCircle[2][1]{
  \LatitudePlane{\angEl}{#2}
  \tikzset{current plane/.prefix style={scale=#1}}
  \pgfmathsetmacro\sinVis{sin(#2)/cos(#2)*sin(\angEl)/cos(\angEl)}
  \pgfmathsetmacro\angVis{asin(min(1,max(\sinVis,-1)))}
  \draw[current plane] (\angVis:1) arc (\angVis:-\angVis-180:1);
  \draw[current plane,dashed] (180-\angVis:1) arc (180-\angVis:\angVis:1);
}

\tikzset{%
  >=latex, 
  inner sep=0pt,%
  outer sep=2pt,%
  mark coordinate/.style={inner sep=0pt,outer sep=0pt,minimum size=3pt,
    fill=black,circle}%
}

\newcommand{\EquatorialBloch}{\begin{tikzpicture}[scale = 0.85]

\def\R{3} 
\def\angEl{25} 
\def\angAz{-100} 
\def\angPhiOne{-55} 
\def\angPhiTwo{-20} 
\def\angBeta{0} 

\pgfmathsetmacro\H{\R*cos(\angEl)} 
\LongitudePlane[xzplane]{\angEl}{\angAz}
\LongitudePlane[pzplane]{\angEl}{\angPhiOne}
\LongitudePlane[qzplane]{\angEl}{\angPhiTwo}
\LatitudePlane[equator]{\angEl}{0}

\fill[ball color=white] (0,0) circle (\R); 
\draw (0,0) circle (\R);

\coordinate (O) at (0,0);
\coordinate (N) at (0,\H);
\coordinate (S) at (0,-\H);
\path[xzplane] (\R,0) coordinate (XE);
\path[pzplane] (\angBeta:\R) coordinate (P);
\path[pzplane] (\R,0) coordinate (PE);
\path[qzplane] (\angBeta:\R) coordinate (Q);
\path[qzplane] (\R,0) coordinate (QE);

\DrawLatitudeCircle[\R]{0}; 

\node[above=8pt] at (N) {$|0\rangle$};
\node[below=8pt] at (S) {$|1\rangle$};
\draw[dashed] (S) -- (N);
\draw[dashed,->] (O) -- (PE);
\draw[dashed,->] (O) -- (XE);
\draw[dashed,->] (O) -- (QE);
\path[pzplane] (0.5*\angBeta:\R) ;
\path[qzplane] (0.5*\angBeta:\R) ;
\draw[equator,->,thin] (\angAz:0.5*\R) to[bend right=30]
    node[pos=0.4,above] {$\phi_1$} (\angPhiOne:0.5*\R);
\draw[equator,->,thin] (\angAz:0.7*\R) to[bend right=40]
    node[pos=0.6,above] {$\phi_2$} (\angPhiTwo:0.7*\R);
\draw[qzplane,->,thin] (90:0.3*\R) to[bend left=45]
    node[midway,right] {$\theta = \frac{\pi}{2}$} (\angBeta:0.3*\R);

\draw[equator] (-180:\R) arc (-180:0:\R) ;
\draw[equator,dashed] (0:\R) arc (0:180:\R) ;

\path[xzplane] (0:\R) node[below left] {$\phi = 0$};
\path[xzplane] (\angBeta:\R);

  \fill[black]   (O) circle(2pt) ;
  \fill[black]   (N) circle(2pt) ;
  \fill[black]   (S) circle(2pt) ;

\end{tikzpicture}}

\newcommand{\ABTableNUM}[4]{
\begin{array}{cc|cccc}
  \text{A} & \text{B} & 0\,0 & 0\,1 & 1\,0 & 1\,1 \\
  \hline
  a_1 & b_1 & #1 \\ 
  a_1 & b_2 & #2 \\ 
  a_2 & b_1 & #3 \\ 
  a_2 & b_2 & #4 \\ 
\end{array}
}

\newcommand{\ABTableBellmodel}{
\ABTableNUM
{\nfrac{1}{2} &  0  &  0  & \nfrac{1}{2}}
{\nfrac{3}{8} & \nfrac{1}{8} & \nfrac{1}{8} & \nfrac{3}{8}}
{\nfrac{3}{8} & \nfrac{1}{8} & \nfrac{1}{8} & \nfrac{3}{8}}
{\nfrac{1}{8} & \nfrac{3}{8} & \nfrac{3}{8} & \nfrac{1}{8}}}

\newcommand{\ABTablePRbox}{
\ABTableNUM
{\nfrac{1}{2} &      0       &       0       & \nfrac{1}{2}}
{\nfrac{1}{2} &      0       &       0       & \nfrac{1}{2}}
{\nfrac{1}{2} &      0       &       0       & \nfrac{1}{2}}
{     0       & \nfrac{1}{2} &  \nfrac{1}{2} &      0      }}

\newcommand{\ABCTableNUM}[8]{
\begin{array}{ccc|cccccccc}
  \text{A} & \text{B} & \text{C} & 0\,0\,0 & 0\,0\,1 & 0\,1\,0 & 0\,1\,1 & 1\,0\,0 & 1\,0\,1 & 1\,1\,0 & 1\,1\,1\\
  \hline
  a_1 & b_1 & c_1 & #1 \\ 
  a_1 & b_1 & c_2 & #2 \\ 
  a_1 & b_2 & c_1 & #3 \\ 
  a_1 & b_2 & c_2 & #4 \\ 
  a_2 & b_1 & c_1 & #5 \\ 
  a_2 & b_1 & c_2 & #6 \\ 
  a_2 & b_2 & c_1 & #7 \\ 
  a_2 & b_2 & c_2 & #8 \\ 
\end{array}
}

\newcommand{\ABCTableGHZ}{
\ABCTableNUM
{     0       & \nfrac{1}{4} & \nfrac{1}{4}  &      0       & \nfrac{1}{4} &      0       &      0       & \nfrac{1}{4}}
{\nfrac{1}{8} & \nfrac{1}{8} & \nfrac{1}{8}  & \nfrac{1}{8} & \nfrac{1}{8} & \nfrac{1}{8} & \nfrac{1}{8} & \nfrac{1}{8}}
{\nfrac{1}{8} & \nfrac{1}{8} & \nfrac{1}{8}  & \nfrac{1}{8} & \nfrac{1}{8} & \nfrac{1}{8} & \nfrac{1}{8} & \nfrac{1}{8}}
{\nfrac{1}{4} &      0       &       0       & \nfrac{1}{4} &      0       & \nfrac{1}{4} & \nfrac{1}{4} &      0      }
{\nfrac{1}{8} & \nfrac{1}{8} & \nfrac{1}{8}  & \nfrac{1}{8} & \nfrac{1}{8} & \nfrac{1}{8} & \nfrac{1}{8} & \nfrac{1}{8}}
{\nfrac{1}{4} &      0       &       0       & \nfrac{1}{4} &      0       & \nfrac{1}{4} & \nfrac{1}{4} &      0      }
{\nfrac{1}{4} &      0       &       0       & \nfrac{1}{4} &      0       & \nfrac{1}{4} & \nfrac{1}{4} &      0      }
{\nfrac{1}{8} & \nfrac{1}{8} & \nfrac{1}{8}  & \nfrac{1}{8} & \nfrac{1}{8} & \nfrac{1}{8} & \nfrac{1}{8} & \nfrac{1}{8}}}

\newcommand{\ABCTableLocalGHZ}{
\ABCTableNUM
{\nfrac{1}{16}& \nfrac{3}{16}& \nfrac{3}{16}& \nfrac{1}{16}& \nfrac{3}{16}& \nfrac{1}{16}& \nfrac{1}{16}& \nfrac{3}{16}}
{\nfrac{1}{8} & \nfrac{1}{8} & \nfrac{1}{8} & \nfrac{1}{8} & \nfrac{1}{8} & \nfrac{1}{8} & \nfrac{1}{8} & \nfrac{1}{8} }
{\nfrac{1}{8} & \nfrac{1}{8} & \nfrac{1}{8} & \nfrac{1}{8} & \nfrac{1}{8} & \nfrac{1}{8} & \nfrac{1}{8} & \nfrac{1}{8} }
{\nfrac{3}{16}& \nfrac{1}{16}& \nfrac{1}{16}& \nfrac{3}{16}& \nfrac{1}{16}& \nfrac{3}{16}& \nfrac{3}{16}& \nfrac{1}{16}}
{\nfrac{1}{8} & \nfrac{1}{8} & \nfrac{1}{8} & \nfrac{1}{8} & \nfrac{1}{8} & \nfrac{1}{8} & \nfrac{1}{8} & \nfrac{1}{8} }
{\nfrac{3}{16}& \nfrac{1}{16}& \nfrac{1}{16}& \nfrac{3}{16}& \nfrac{1}{16}& \nfrac{3}{16}& \nfrac{3}{16}& \nfrac{1}{16}}
{\nfrac{3}{16}& \nfrac{1}{16}& \nfrac{1}{16}& \nfrac{3}{16}& \nfrac{1}{16}& \nfrac{3}{16}& \nfrac{3}{16}& \nfrac{1}{16}}
{\nfrac{1}{8} & \nfrac{1}{8} & \nfrac{1}{8} & \nfrac{1}{8} & \nfrac{1}{8} & \nfrac{1}{8} & \nfrac{1}{8} & \nfrac{1}{8} }}

\newcommand{\ABCTableOtherLocal}{
\ABCTableNUM
{\nfrac{1}{16}& \nfrac{5}{16}& \nfrac{5}{16}& \nfrac{1}{16}& \nfrac{1}{16}& \nfrac{1}{16}& \nfrac{1}{16}& \nfrac{1}{16}}
{\nfrac{1}{4} & \nfrac{1}{8} & \nfrac{1}{8} & \nfrac{1}{4} & \nfrac{1}{8} &      0       &      0       & \nfrac{1}{8} }
{\nfrac{1}{4} & \nfrac{1}{8} & \nfrac{1}{8} & \nfrac{1}{4} & \nfrac{1}{8} &      0       &      0       & \nfrac{1}{8} }
{\nfrac{5}{16}& \nfrac{1}{16}& \nfrac{1}{16}& \nfrac{5}{16}& \nfrac{1}{16}& \nfrac{1}{16}& \nfrac{1}{16}& \nfrac{1}{16}}
{\nfrac{1}{8} & \nfrac{1}{4} & \nfrac{1}{4} & \nfrac{1}{8} &      0       & \nfrac{1}{8} & \nfrac{1}{8} &      0       }
{\nfrac{5}{16}& \nfrac{1}{16}& \nfrac{1}{16}& \nfrac{5}{16}& \nfrac{1}{16}& \nfrac{1}{16}& \nfrac{1}{16}& \nfrac{1}{16}}
{\nfrac{5}{16}& \nfrac{1}{16}& \nfrac{1}{16}& \nfrac{5}{16}& \nfrac{1}{16}& \nfrac{1}{16}& \nfrac{1}{16}& \nfrac{1}{16}}
{\nfrac{1}{4} & \nfrac{1}{8} & \nfrac{1}{8} & \nfrac{1}{4} & \nfrac{1}{8} &      0       &      0       & \nfrac{1}{8} }}

\newcommand{\ABCTableOtherStronglyNL}{
\ABCTableNUM
{     0       & \nfrac{1}{8} & \nfrac{1}{8}  &      0       & \nfrac{3}{8} &      0       &      0       & \nfrac{3}{8}}
{     0       & \nfrac{1}{8} & \nfrac{1}{8}  &      0       & \nfrac{1}{8} & \nfrac{1}{4} & \nfrac{1}{4} & \nfrac{1}{8}}
{     0       & \nfrac{1}{8} & \nfrac{1}{8}  &      0       & \nfrac{1}{8} & \nfrac{1}{4} & \nfrac{1}{4} & \nfrac{1}{8}}
{\nfrac{1}{8} &      0       &      0        & \nfrac{1}{8} &      0       & \nfrac{3}{8} & \nfrac{3}{8} &      0      }
{\nfrac{1}{8} &      0       &      0        & \nfrac{1}{8} & \nfrac{1}{4} & \nfrac{1}{8} & \nfrac{1}{8} & \nfrac{1}{4}}
{\nfrac{1}{8} &      0       &      0        & \nfrac{1}{8} &      0       & \nfrac{3}{8} & \nfrac{3}{8} &      0      }
{\nfrac{1}{8} &      0       &      0        & \nfrac{1}{8} &      0       & \nfrac{3}{8} & \nfrac{3}{8} &      0      }
{     0       & \nfrac{1}{8} & \nfrac{1}{8}  &      0       & \nfrac{1}{8} & \nfrac{1}{4} & \nfrac{1}{4} & \nfrac{1}{8}}}

\newcommand{\ie}{i.e.~}
\newcommand{\eg}{e.g.~}

\newcommand{\thmCFmonotoneSTATEMENT}{The contextual fraction is
invariant under relabellings,
and non-increasing under translation of measurements and coarse-graining of outcomes.
For the combining operations, it satisfies the following properties:
\begin{itemize}
\item
$\CF(\lambda e + (1-\lambda) e') \leq \lambda\CF(e) + (1-\lambda) \CF(e)$
\item
$\CF(e \choice e') = \max\enset{\CF(e),\CF(e')}$
\item
$\CF(e \otimes e') = \CF(e) + \CF(e') - \CF(e)\CF(e')$
\end{itemize}}

\newif\iffullform

\newcommand{\thmCFBellineqsSTATEMENT}{
Let $e$ be an empirical model.
\begin{enumerate*}[label=(\roman*)]
\item\label{item:atmost} The normalised violation by $e$ of any Bell inequality is at most $\CF(e)$;
\item\label{item:ineq} if $\CF(e) > 0$, this bound is attained, \ie there exists a Bell inequality whose normalised violation by $e$ is $\CF(e)$;
\item\label{item:tight} moreover, for any decomposition of the form
\iffullform
$e = \NCF(e) e^{NC} + \CF(e)e^{SC}$,
\else
 \eqref{eq:decomp},
\fi
this Bell inequality is tight at the non-contextual model $e^{NC}$ (provided $\NCF(e)>0$)
and maximally violated at the strongly contextual model $e^{SC}$.
\end{enumerate*}
}

\newcommand{\thmCFMBQCSTATEMENT}{
Let $\fdec{f}{2^m}{2^l}$ be a Boolean function and consider an \ltMBQC\ that uses the empirical model $e$
to compute $f$ with average success probability $\bar{p}_S$ over all $2^m$ possible inputs, and corresponding average failure probability $\bar{p}_F = 1 - \bar{p}_S$. Then, $\bar{p}_F \; \geq \; \NCF(e)\tilde{\nu}(f)$.}

\newcommand{\thmCFmonotoneMAIN}{\ifbothtexts
\begin{restatable}{theorem}{thmCFmonotone}\label{thm:CFmonotone}
\thmCFmonotoneSTATEMENT
\end{restatable}
\else
\begin{theorem}\label{thm:CFmonotone}
\thmCFmonotoneSTATEMENT
\end{theorem}
\fi
}

\newcommand{\thmCFmonotoneSUPP}{\ifbothtexts
\thmCFmonotone*
\else
\begin{theorem}\label{thm:CFmonotone}
\thmCFmonotoneSTATEMENT
\end{theorem}
\fi
}

\newcommand{\thmCFBellineqsMAIN}{\ifbothtexts
\begin{restatable}{theorem}{thmCFBellineqs}\label{thm:CFBellineqs}
\thmCFBellineqsSTATEMENT
\end{restatable}
\else
\begin{theorem}\label{thm:CFBellineqs}
\thmCFBellineqsSTATEMENT
\end{theorem}
\fi
}

\newcommand{\thmCFBellineqsSUPP}{\ifbothtexts
\thmCFBellineqs*
\else
\begin{theorem}\label{thm:CFBellineqs}
\thmCFBellineqsSTATEMENT
\end{theorem}
\fi
}

\newcommand{\thmCFMBQCMAIN}{\ifbothtexts
\begin{restatable}{theorem}{thmCFMBQC}\label{thm:CFMBQC}
\thmCFMBQCSTATEMENT
\end{restatable}
\else
\begin{theorem}\label{thm:CFMBQC}
\thmCFMBQCSTATEMENT
\end{theorem}
\fi
}

\newcommand{\thmCFMBQCSUPP}{\ifbothtexts
\thmCFMBQC*
\else
\begin{theorem}\label{thm:CFMBQC}
\thmCFMBQCSTATEMENT
\end{theorem}
\fi
}

\begin{document}



\let\reusablemaketitle\maketitle
\let\reusableauthor\author
\let\reusableaffiliation\affiliation

\newcommand{\affilone}{Department of Computer Science, University of Oxford \\ Wolfson Building, Parks Rd, Oxford OX1 3QD, United Kingdom}
\newcommand{\affiltwo}{School of Informatics, University of Edinburgh \\ Informatics Forum, 10 Crichton Street, Edinburgh EH8 9AB, United Kingdom}

\ifmaintext
\title{The Contextual Fraction as a Measure of Contextuality}
\author{Samson Abramsky}%
\affiliation{\affilone}%
\author{Rui Soares Barbosa}%
\affiliation{\affilone}%
\author{Shane Mansfield}%
\affiliation{\affiltwo}%
\date{\today}

\begin{abstract}
We consider the contextual fraction as a quantitative measure of contextuality of empirical models, \ie tables of probabilities of measurement outcomes in an experimental scenario.
It provides a general way to compare the degree of contextuality across measurement scenarios; it bears a precise relationship to violations of Bell inequalities; its value, and a witnessing inequality, can be computed using linear programming; it is monotone with respect to  the ``free'' operations of a resource theory for contextuality; and it measures quantifiable advantages in informatic tasks, such as games and a form of measurement based quantum computing.
\end{abstract}

\fi
\ifmaintext
\maketitle


\prlsection{Introduction}%
Recent results have established the r\^ole of contextuality as a resource for increasing the computational power of specific models of computation \cite{AndersBrowne,RaussendorfSC}, including enabling universal quantum computation \cite{Howard2014}.
From this perspective, it is particularly relevant to look for appropriate measures of contextuality, and indeed to pose the question of what constitutes a good measure.

In this Letter, we propose a measure of contextuality---the contextual fraction---which provides a quantitative grading between non-contextuality, at one extreme, and maximal contextuality, at the other. 
A maximally contextual empirical model is one that admits no proper decomposition into a convex combination of a non-contextual model and another model.
In this sense, it is meaningful to consider both the non-contextual and contextual fractions of any no-signalling empirical model.

These definitions are made in the general setting of the approach to contextuality introduced in \cite{AbramskyBrandenburger},
in which nonlocality is seen as a special case of contextuality.

We show that the contextual fraction has a number of desirable properties:
\begin{enumerate*}[label=(\roman*)]
\item it is fully general in the sense that it applies in any measurement scenario;
\item it is bounded and normalised, taking values in the interval $[0,1]$, with $0$ indicating non-contextuality and $1$ indicating strong contextuality, so it may be used to sensibly compare the degree of contextuality of empirical models not just in a given measurement scenario but also across scenarios;
\item it has a precise relationship with violations of  Bell inequalities, being the maximum normalised violation attained by the  empirical model for any  Bell inequality on the corresponding measurement scenario;
\item\label{lp} both the contextual fraction and a witnessing  Bell inequality are computable using linear programming---these were implemented and used for computational exploration of some quantum examples;
\item\label{rt} it is monotone with respect to a range of operations on empirical models that intuitively do not generate contextuality, and thus constitute natural ``free'' operations in a resource theory of contextuality, analogous to the resource theory of entanglement under LOCC operations \cite{HorodeckiEntanglement}, and subsuming existing resource theories for nonlocality \cite{BarrettEtAl:NonlocalCorrelationsResource,GallegoEtAl:OperationalFrameworkNonlocality,DeVincenteNL};
\item\label{adv} finally, it is related to a quantifiable increase of computational power in a certain form of measurement-based quantum computation, sharpening the results of \cite{RaussendorfSC}, and similarly to advantage in games.
\end{enumerate*}

We leave for future work an analysis of the relationship between the contextual fraction and other possible measures \footnote{These include a negative probability measure introduced in \cite{AbramskyBrandenburger}, contextuality-by-default measures \cite{dzhafarov2015contextuality}, and various other measures considered for nonlocality in specific measurement scenarios, \eg noise-based \cite{kaszlikowski2000violations,acin2002quantum,junge2010operator} and inefficiency-based \cite{massar2002bell,laplante2012classical,degorre2009communication} measures, known to relate to communication complexity \cite{degorre2009communication,laplante2016robust}.}, and further development of \ref{adv}.

\prlsection{General framework for contextuality}%
We briefly summarise the framework introduced in \cite{AbramskyBrandenburger}.
The main objects of study are \emph{empirical models}: tables of data, specifying probability distributions over the joint outcomes of sets of compatible measurements.

A \emph{measurement scenario} is an abstract description of a particular experimental setup.
It consists of a triple $\XMO$ where:
$X$ is a finite set of measurements;
$O$ is a finite set of outcome values for each measurement;
and $\M$ is a set of subsets of $X$.
Each $C \in \M$ is called a  \emph{measurement context}, and represents a set of measurements that can be performed together.


Examples of measurement scenarios include multipartite Bell-type scenarios familiar from discussions of nonlocality, Kochen--Specker configurations, measurement scenarios associated with qudit stabiliser quantum mechanics, and more.
For example, the well-known $(2,2,2)$ Bell scenario,
where two experimenters, Alice and Bob,
can each choose between performing one of two different measurements, say $a_1$ or $a_2$ for Alice and $b_1$ or $b_2$ for Bob,
obtaining one of two possible outcomes, is represented as follows:
\begin{align*}
X &= \enset{a_1,a_2,b_1,b_2} \qquad\qquad O = \enset{0,1} \\
\M &= \enset{ \enset{a_1,b_1} , \enset{a_1,b_2} , \enset{a_2,b_1}, \enset{a_2, b_2} }
\Mdot
\end{align*}

Given this description of the experimental setup, then either performing repeated runs of such experiments with varying choices of measurement context and recording the frequencies of the various outcome events, or calculating theoretical predictions for the probabilities of these outcomes, results in a
probability table as in Table~\ref{tab:BellPRmodel}.
\begin{table}
\caption{\label{tab:BellPRmodel} Two empirical models on the $(2,2,2)$ Bell scenario:
the well-known CHSH model \cite{CHSH}, obtained from local projective measurements equatorial at angles $0$ (for $a_1$, $b_1$) and $\nfrac{\pi}{3}$ (for $a_2$, $b_2$)  on the maximally entangled two-qubit Bell-state $\BellState = \BellStateExpandLine$; and the Popescu--Rohrlich box.
}
$\ABTableBellmodel \quad\quad \ABTablePRbox$
\end{table}

Such data is formalised as an \emph{empirical model} for the given measurement scenario $\XMO$.
For each measurement context $C$, there is a probability distribution $e_C$ on the joint outcomes of performing all the measurements in $C$; that is, on the set 
$O^C$ of functions assigning an outcome in $O$ to each measurement in $C$.

We require that the marginals of these distributions agree whenever contexts overlap, \ie
\[\Forall{C, C' \in \M} e_C|_{C \cap C'} = e_{C'}|_{C \cap C'} \Mcomma\]
where the notation $e_C|_{U}$ with $U \subseteq C$ stands for marginalisation of probability distributions (to `forget' the outcomes of some measurements): for $t \in O^U$, $e_C|_{U}(t) \defeq \sum_{s \in O^C, s|_U = t}e_C(s)$.
The requirement of \emph{compatibility of marginals} is a generalisation of the usual \emph{no-signalling} condition, and is satisfied in particular by all empirical models arising from quantum predictions \cite{AbramskyBrandenburger}.


An empirical model is said to be \emph{contextual} if this family of distributions
cannot itself be obtained as the marginals of a single probability distribution on global assignments of outcomes to all measurements,
\ie a distribution $d$ on $O^X$ (where $O^X$ acts as a canonical set of deterministic hidden variables) such that $\Forall{C \in \M} d|_C = e_C$.
Equivalently \cite{AbramskyBrandenburger}, contextual empirical models are those which have no realisation by factorisable hidden variable models; thus for Bell-type measurement scenarios contextuality specialises to the usual notion of \emph{nonlocality}.

In certain cases, one can witness contextuality from merely the possibilistic, rather than probabilistic, information contained in an empirical model---\ie which events are {possible} (with non-zero probability) and which are impossible (with zero probability).
A yet stronger form of contextuality occurs when no global assignment of outcomes is even consistent with the possible events:
an empirical model $e$ is said to be \emph{strongly contextual} if there is no global assignment $g \in O^X$ such that  $\Forall{C \in \M} e_C(g|_C) > 0$.
An example is the Popescu--Rohrlich box (Table~\ref{tab:BellPRmodel}).
This is the highest level in the qualitative hierarchy of strengths of contextuality introduced in \cite{AbramskyBrandenburger}.

\prlsection{The contextual fraction}%
Given two empirical models $e$ and $e'$ on the same measurement scenario
and $\lambda \in [0,1]$, we define the empirical model $\lambda e + (1-\lambda) e'$ by taking the convex sum of probability distributions at each context.
Compatibility is preserved by this convex sum, hence it yields a well-defined empirical model.
%

A natural question to ask is: what fraction of a given empirical model $e$ admits a non-contextual explanation?
This approach enables a refinement of the binary notion of contextuality \textit{vs} non-contextuality into a quantitative grading.
Instead of asking for a probability distribution on global assignments
that marginalises to the empirical distributions at each context,
we ask only for a subprobability distribution \footnote{A subprobability distribution on a set $S$ is a map $\fdec{b}{S}{\RRpz}$ with finite support and $\weight{b} \leq 1$, where $\weight{b} \defeq \sum_{s \in S}b(s)$ is called its weight. The set of subprobability distributions on $S$ is ordered pointwise: $b'$ is a subdistribution of $b$ (written $b' \leq b$) whenever $\Forall{s \in S}b'(s) \leq b(s)$.} $b$ on global assignments $O^X$
that marginalises at each context to a subdistribution of the empirical data, thus explaining a fraction of the events,
\ie $\Forall{C \in \M} b|_C \leq e_C$.
Equivalently, we ask for a convex decomposition
\begin{equation}\label{eq:convexdec}
  e = \lambda e^{NC} + (1-\lambda) e'
\end{equation}
where $e^{NC}$ is a non-contextual model and $e'$ is another (no-signalling) empirical model.
The maximum 
weight of such a global subprobability distribution, or the maximum possible value of $\lambda$ in such a decomposition 
is called the \emph{non-contextual fraction} of $e$,
and generalises the  \emph{local fraction} previously introduced for models on Bell-type scenarios
\cite{ElitzurPopescuRohrlich1992:QuantumNonlocalityForEachPairInAnEnsemble} \footnote{See also \cite{BarretKentPironio2006:MaximallyNonlocalAndMonogamousQuantumCorrelations, AolitaEtAl2012:FullyNonlocalQuantumCorrelations}
where the term \stress{local fraction} is actually used.}.
We denote it by $\NCF(e)$, and the contextual fraction by $\CF(e) \defeq 1 - \NCF(e)$.

The notion of contextual fraction in general scenarios was introduced in \cite{AbramskyBrandenburger},
where it was proved that a model is strongly contextual if and only if its contextual fraction is $1$.
In fact, in any convex decomposition of the form \eqref{eq:convexdec} giving maximal weight to the non-contextual model, the other model will necessarily be strongly contextual.
This means that any empirical model $e$ admits a convex decomposition
\begin{equation}\label{eq:decomp}
e = \NCF(e) \, e^{NC} + \CF(e) \, e^{SC}
\end{equation}
into a non-contextual and a strongly contextual model.
Note that $e^{NC}$ and $e^{SC}$ are not necessarily unique.

\prlsection{Computing the contextual fraction via LP}%
The task of finding a consistent probability subdistribution with maximum weight for a given empirical model can be formulated as a linear programming problem.
This is a relaxation of the test for contextuality
by solving a system of linear equations over the nonnegative reals from \cite{AbramskyBrandenburger}.

Fix a measurement scenario $\XMO$.
Let $n \defeq |O^X|$ be the number of global assignments $g$,
and $m \defeq \sum_{C \in \M} |O^C| = |\setdef{\tuple{C,s}}{C\in\M, s \in O^C}|$ be the number of local assignments ranging over contexts.
The \emph{incidence matrix} \cite{AbramskyBrandenburger} $\IM$ is an $m \times n$ $(0,1)$-matrix that records the restriction relation between global and local assignments:
\[
\IM[\tuple{C,s},g] \defeq
\begin{cases}
 1 \;\text{ if $g|_C = s$;}
\\ 
 0 \;\text{ otherwise.}
\end{cases}
\]

An empirical model $e$ can be represented as a vector $\ve \in \RR^m$, with the component $\ve[\tuple{C,s}]$ recording the probability given by the model to the assignment $s$ at the measurement context $C$, $e_C(s)$. 
This vector is a flattened version of the table used to represent the empirical model (e.g.~Table~\ref{tab:BellPRmodel}).
The columns of the incidence matrix, $\mathbf{M}[-,g]$, are the vectors corresponding to the (non-contextual) deterministic models obtained from global assignments $g \in O^X$.
Recall that every non-contextual model can be written as a mixture of these.
A probability distribution on global assignments can be represented as a vector $\dd \in \RR^n$ with non-negative components summing to $1$, and then the corresponding non-contextual model is represented by the vector $\IM \, \dd$.
So a model $e$ is non-contextual if and only if there exists a $\dd \in \RR^n$ such that:
\[\IM \, \dd \,=\, \ve \Mand \dd \geq \mathbf{0} \Mdot\]
Note that the first condition implies that $\dd$ is normalised.


A global subprobability distribution is also represented by a vector $\bb \in \RR^n$ with non-negative components, its weight being given by the dot product $\vone \cdot \bb$,
where $\vone \in \RR^n$ is the vector whose $n$ components are each $1$. The following LP thus calculates the non-contextual fraction of an empirical model $e$,
with $\NCF(e) = \vone\cdot \bb^*$ where $\bb^*$ is an optimal solution:
\begin{equation}\label{LP:quantifying}
\begin{alignedat}{3}
&\textLP{Find }       \;\;&&  \bb \in \RR^n
\\
&\textLP{maximising }   \;\;&& \vone \cdot \bb
\\
&\textLP{subject to } \;\;&& \IM \, \bb \,\leq\, \ve
\\
&\textLP{and }        \;\;&& \bb \,\geq\, \vzero
& \Mdot
\end{alignedat}
\end{equation}

\prlsection{Violations of generalised Bell inequalities}%
We now provide further justification for viewing the contextual fraction as a measure of contextuality by relating it to violations of contextuality-witnessing inequalities.

An \emph{inequality} for a scenario $\tuple{X,\M,O}$ is given by
a vector $\av \in \RR^m$ of real coefficients indexed by local assignments
$\tuple{C,s}$, 
and a bound $R$.
For a model $e$, the inequality reads
$\av \cdot \ve  \, \leq \, R$,
where
\[\av \cdot \ve \; = \; \sum_{C \in \M, s \in O^C} \, \av[\tuple{C,s}] \, e_C(s) \Mdot \]
Without loss of generality, we can take $R$ to be non-negative (in fact, even $R = 0$)
as any inequality is equivalent to one of this form.
We call it a \emph{Bell inequality} if it is satisfied by every non-contextual model.
This generalises the usual notion of Bell inequality, which is defined for Bell-type scenarios for nonlocality, to apply to any contextuality scenario.
If, moreover, it is saturated by some non-contextual model, the Bell inequality is said to be \emph{tight}.
A Bell inequality establishes a bound for the value of $\av \cdot \ve$ amongst non-contextual models $e$.
For more general models, this quantity is limited only
by the algebraic bound 
\[\|\av\| \defeq \sum_{C \in \M}\max\setdef{\av[\tuple{C,s}]}{s\in O^C} \Mdot\]
Note that we will consider only inequalities satisfying $R < \|\av\|$, which excludes inequalities trivially satisfied by all models, and avoids cluttering the presentation with special caveats about division by 0.

The \emph{violation} of a Bell inequality $\tuple{\av,R}$ by a model $e$ is 
$\max\enset{0,\av \cdot \ve - R}$.
However, it is useful to normalise this value by the maximum possible violation
in order to give a better idea of the \stress{extent} to which the model violates the inequality.
The \emph{normalised violation}
of the Bell inequality by the model $e$ is 
\[\frac{\max\enset{0,\av \cdot \ve - R}}{\|\av\| - R} \Mdot\]

\fullformfalse
\thmCFBellineqsMAIN

The proof of this result is based on the Strong Duality theorem of linear programming \cite{DantzigThapa:LinearProgramming2}.
It provides an LP method of calculating a witnessing Bell inequality for any empirical model $e$.
The symmetric dual of \eqref{LP:quantifying} is the following LP:
\begin{equation}\label{LP:dualquantifying}
\begin{alignedat}{3}
&\textLP{Find }       \;\;&&  \yy \in \RR^m
\\
&\textLP{minimising }   \;\;&& \yy \cdot \ve
\\
&\textLP{subject to } \;\;&& \IM^T \, \yy \,\geq\, \mathbf{1}
\\
&\textLP{and }        \;\;&& \yy \,\geq\, \mathbf{0}
& \Mdot
\end{alignedat}
\end{equation}
The  Strong Duality theorem says that, if $\bb^*$ is a solution for \eqref{LP:quantifying}, then there is a solution $\yy^*$ for \eqref{LP:dualquantifying} satisfying
\begin{equation}
\label{sdeq}
 \vone \cdot \bb^* \; = \; \yy^* \cdot \ve . 
 \end{equation}
Defining  $\av^* \defeq |\M|\inv \vone -  \yy^*$, one can show using \eqref{sdeq} that
the Bell inequality determined by
$\av^*$ as the vector of coefficients and with bound $R = 0$
satisfies parts \ref{item:ineq} and \ref{item:tight} of the Theorem.
A detailed proof is provided in the supplemental material \footnote{See Supplemental Material at [URL will be inserted by publisher] for rigorous proofs and further details on our computational explorations.}.

\prlsection{Monotonicity}%
A key desideratum of a useful measure of contextuality is that it be a monotone for the free operations of a resource theory for contextuality.
A fuller treatment of this subject will be presented in a forthcoming article by the authors;
here, we consider the properties of the contextual fraction with respect to some of these operations.

We consider the following operations:  first, translation of measurements (including restriction to a smaller set of measurements, replication of measurements, etc.); secondly, coarse-graining of outcomes. Special cases of these give isomorphic relabelling of measurements and outcomes.
We also consider operations that combine two empirical models to build a new one.
The first of these is \emph{probabilistic mixing} with a weight $\lambda \in [0,1]$.
The second is \emph{controlled choice}:
given empirical models $e$ and $e'$ on scenarios $\XMO$ and $\tuple{X',\M',O}$ respectively,
$e \choice e'$ is defined on the scenario $\tuple{X \sqcup X',\M \sqcup \M',O}$
by $(e \choice e')_C := e_C$ for $C \in \M$ and $(e \choice e')_{C'} := e_{C'}$ for $C' \in \M'$. 
The third is a \emph{product}:
$e \otimes e'$ is an empirical model defined on the scenario $\tuple{X \sqcup X',\M \star \M',O}$, where
$\M \star \M' := \setdef{C \sqcup C'}{C \in \M, C' \in \M'}$, by $(e \otimes e')_{C\sqcup C'}\tuple{s,s'} := e_C(s) \, e'_{C'}(s')$ for all $C \in M$, $C' \in \M'$, $s \in O^C$, and $s' \in O^{C'}$.

These operations
can be used to construct any local empirical model on Bell scenarios
starting from a very simple ``generator'': a deterministic model over a single measurement.
This is illustrated in the supplemental material.

%
%
%

\thmCFmonotoneMAIN
A consequence of this result is that, for any of the combining operations, when $e'$ is a non-contextual model (and thus composing with $e'$ is a free operation),
$\CF$ acts as a monotone: the contextual fraction of the new model is at most that of $e$ (in fact, with equality holding for both choice and product).

\prlsection{Computational explorations}%
General computational tools in the form of a \textit{Mathematica} package have been developed implementing the two LPs described above to calculate,
for \stress{any} empirical model in \stress{any} scenario:
the (non-)contextual fraction, a decomposition of the form (2), and the generalised Bell inequality from Theorem 1-(ii) for which the maximal violation is achieved.
The package also calculates quantum empirical models from any (pure or mixed) state and any specified sets of compatible measurements.

As an example to illustrate the use of this package, we consider the empirical models obtained
from local measurements on various $n$-qubit states. On each qubit, we allow the same two local measurements, equatorial on the Bloch sphere, parametrised by angles $\phi_1$ and $\phi_2$.
Figure~\ref{fig:plots} plots the contextual fraction of the resulting models as a function of these angles.

\begin{figure*}
\begin{center}
\begin{tabular}{ccc}
\includegraphics[width=0.32\textwidth]{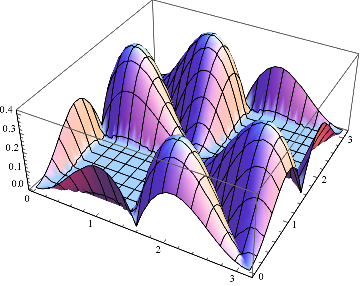} 
&
\includegraphics[width=0.32\textwidth]{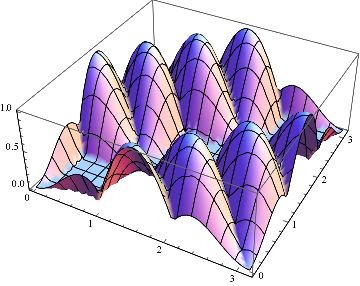}
&
\includegraphics[width=0.32\textwidth]{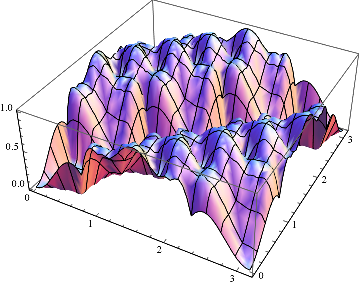}
\\
(a) & (b) & (c)
\end{tabular}
\end{center}
\caption{\label{fig:plots} Plots of the contextual fraction for empirical models obtained with projective measurements at $\phi_1$ and $\phi_2$ in the $X$--$Y$ plane for each qubit on the states: (a) the Bell state $\BellState \defeq \BellStateExpand$; (b) $\GHZStateN{3}$; (c) $\GHZStateN{4}$, where the $n$-partite GHZ state ($n > 2$) is given by $\GHZStateNN = \GHZStateNExpand$.
}
\end{figure*}

Computational explorations of this kind can be a useful tool for guiding research, pointing the way to conjectures and results (\eg \cite{AbramskyConstantinYing2015:HardyIsAlmostEverywhere,abramsky2014classification,Mansfield:DPhil-thesis}).
A more detailed analysis of the examples from Figure \ref{fig:plots},
leading to
the characterisation of a family of strongly contextual models arising from $n$-partite GHZ states,
can be found in the supplemental material.


\prlsection{Applications to quantum computation}%
Contextuality has been associated with quantum advantage in certain information-processing and computational tasks.
One use for a measure of contextuality is to quantify such advantages.

One computational model in which contextuality has been associated with an advantage is
measurement-based quantum computation (MBQC). An \ltMBQC\ is a process with $m$ classical bits of input and $l$ of output, using an $(n,2,2)$ empirical model ($n$ parties, $2$ measurement settings per party, $2$ outcomes per measurement) as a resource.
The classical control---which pre-processes the inputs, determines the flow of measurements by choosing which sites to measure next and with which measurement setting (potentially depending on previous outcomes), and post-processes to produce the outputs---can only perform $\ZZ_2$-linear computations.
The additional power to compute non-linear functions thus resides in certain resource empirical models.

In \cite[Theorem 2]{RaussendorfSC} it was shown that if an \ltMBQC\ process \stress{deterministically}
calculates a non-$\ZZ_2$-linear Boolean function $\fdec{f}{2^m}{2^l}$,
then the resource is necessarily strongly contextual.
A probabilistic version was also obtained in \cite[Theorem 3]{RaussendorfSC}: contextuality must be present
whenever a non-linear function is calculated with a sufficiently large probability of success.
By analysing that proof, we extract a sharpened version of this result
establishing a precise relationship between the hardness (non-linearity) of the function, the probability of success, and the contextual fraction. 

The \emph{average distance} between two Boolean functions $\fdec{f,g}{2^m}{2^l}$ is given by 
$\tilde{d}(f,g) \defeq 2^{-m}|\setdef{\ii \in 2^m}{f(\ii)\neq g(\ii)}|$.
The average distance of $f$ to the closest $\ZZ_2$-linear function is denoted by $\tilde{\nu}(f)$. 


\thmCFMBQCMAIN
Note that for deterministic computation ($\bar{p}_S=1$) of a non-linear function ($\tilde{\nu} > 0$), we require strong contextuality ($\NCF(e) = 0$), recovering the deterministic result in \cite{RaussendorfSC}.
More generally, for a given non-linear function, the higher the desired success probability the larger the contextual fraction must be.
Additional details, including a rigorous presentation and proof, may be found in the supplemental material.

Similar results can be obtained to quantify advantage in games, generalizing XOR games on Bell scenarios.
A game is specified by $n$ boolean formulae, one for each context, which describe the winning condition that the output must satisfy. If the formulae are $k$-consistent, meaning that at most $k$ of them have a joint satisfying assignment, then the \stress{hardness} of the game is measured by $\frac{(n-k)}{n}$.
One can show that $\bar{p}_F \geq \NCF(e) \frac{(n-k)}{n}$, relating the probability of success, the non-contextual fraction, and the hardness of the task.
See \cite{AbramskyHardy:LogicalBellIneqs} for the relation with Bell inequalities, from which a proof of this result follows.
Details are given in the supplemental material, Theorem~4.
Further development of these ideas is a topic for future research.

\prlsection{Acknowledgements}%
This work was done in part while the authors visited the Simons Institute for the Theory of Computing (supported by the Simons Foundation) at the University of California, Berkeley, as participants of the Logical Structures in Computation programme,
and while SM was based at l'Institut de Recherche en Informatique Fondamentale, Universit\'e Paris Diderot -- Paris 7.
Support from the following is also gratefully acknowledged:
EPSRC EP/N018745/1, Contextuality as a Resource in Quantum Computation (SA, RSB);
Fondation Sciences Math\'{e}matiques de Paris, postdoctoral research grant eotpFIELD15RPOMT-FSMP1, Contextual Semantics for Quantum Theory (SM);
the Oxford Martin School James Martin Program on Bio-inspired Quantum Technologies, (SA, SM, RSB);
FCT -- Funda\c{c}\~{a}o para a Ci\^encia e Tecnologia (Portuguese Foundation for Science and Technology), PhD grant SFRH/BD/94945/2013 (RSB).

\bibliography{refs_cf}

\begin{thebibliography}{32}
\expandafter\ifx\csname natexlab\endcsname\relax\def\natexlab#1{#1}\fi
\expandafter\ifx\csname bibnamefont\endcsname\relax
  \def\bibnamefont#1{#1}\fi
\expandafter\ifx\csname bibfnamefont\endcsname\relax
  \def\bibfnamefont#1{#1}\fi
\expandafter\ifx\csname citenamefont\endcsname\relax
  \def\citenamefont#1{#1}\fi
\expandafter\ifx\csname url\endcsname\relax
  \def\url#1{\texttt{#1}}\fi
\expandafter\ifx\csname urlprefix\endcsname\relax\def\urlprefix{URL }\fi
\providecommand{\bibinfo}[2]{#2}
\providecommand{\eprint}[2][]{\url{#2}}

\bibitem[{\citenamefont{Anders and Browne}(2009)}]{AndersBrowne}
\bibinfo{author}{\bibfnamefont{J.}~\bibnamefont{Anders}} \bibnamefont{and}
  \bibinfo{author}{\bibfnamefont{D.~E.} \bibnamefont{Browne}},
  \bibinfo{journal}{Phys. Rev. Lett.} \textbf{\bibinfo{volume}{102}},
  \bibinfo{pages}{050502} (\bibinfo{year}{2009}).

\bibitem[{\citenamefont{Raussendorf}(2013)}]{RaussendorfSC}
\bibinfo{author}{\bibfnamefont{R.}~\bibnamefont{Raussendorf}},
  \bibinfo{journal}{Phys. Rev. A} \textbf{\bibinfo{volume}{88}},
  \bibinfo{pages}{022322} (\bibinfo{year}{2013}).

\bibitem[{\citenamefont{Howard et~al.}(2014)\citenamefont{Howard, Wallman,
  Veitch, and Emerson}}]{Howard2014}
\bibinfo{author}{\bibfnamefont{M.}~\bibnamefont{Howard}},
  \bibinfo{author}{\bibfnamefont{J.}~\bibnamefont{Wallman}},
  \bibinfo{author}{\bibfnamefont{V.}~\bibnamefont{Veitch}}, \bibnamefont{and}
  \bibinfo{author}{\bibfnamefont{J.}~\bibnamefont{Emerson}},
  \bibinfo{journal}{Nature} \textbf{\bibinfo{volume}{510}},
  \bibinfo{pages}{351} (\bibinfo{year}{2014}).

\bibitem[{\citenamefont{Abramsky and
  Brandenburger}(2011)}]{AbramskyBrandenburger}
\bibinfo{author}{\bibfnamefont{S.}~\bibnamefont{Abramsky}} \bibnamefont{and}
  \bibinfo{author}{\bibfnamefont{A.}~\bibnamefont{Brandenburger}},
  \bibinfo{journal}{New J. Phys.} \textbf{\bibinfo{volume}{13}},
  \bibinfo{pages}{113036} (\bibinfo{year}{2011}).

\bibitem[{\citenamefont{Horodecki et~al.}(2009)\citenamefont{Horodecki,
  Horodecki, Horodecki, and Horodecki}}]{HorodeckiEntanglement}
\bibinfo{author}{\bibfnamefont{R.}~\bibnamefont{Horodecki}},
  \bibinfo{author}{\bibfnamefont{P.}~\bibnamefont{Horodecki}},
  \bibinfo{author}{\bibfnamefont{M.}~\bibnamefont{Horodecki}},
  \bibnamefont{and}
  \bibinfo{author}{\bibfnamefont{K.}~\bibnamefont{Horodecki}},
  \bibinfo{journal}{Rev. Mod. Phys.} \textbf{\bibinfo{volume}{81}},
  \bibinfo{pages}{865} (\bibinfo{year}{2009}).

\bibitem[{\citenamefont{Barrett et~al.}(2005)\citenamefont{Barrett, Linden,
  Massar, Pironio, Popescu, and
  Roberts}}]{BarrettEtAl:NonlocalCorrelationsResource}
\bibinfo{author}{\bibfnamefont{J.}~\bibnamefont{Barrett}},
  \bibinfo{author}{\bibfnamefont{N.}~\bibnamefont{Linden}},
  \bibinfo{author}{\bibfnamefont{S.}~\bibnamefont{Massar}},
  \bibinfo{author}{\bibfnamefont{S.}~\bibnamefont{Pironio}},
  \bibinfo{author}{\bibfnamefont{S.}~\bibnamefont{Popescu}}, \bibnamefont{and}
  \bibinfo{author}{\bibfnamefont{D.}~\bibnamefont{Roberts}},
  \bibinfo{journal}{Physical Review A} \textbf{\bibinfo{volume}{71}},
  \bibinfo{pages}{022101} (\bibinfo{year}{2005}).

\bibitem[{\citenamefont{Gallego et~al.}(2012)\citenamefont{Gallego,
  W\"urflinger, Ac\'{\i}n, and
  Navascu\'es}}]{GallegoEtAl:OperationalFrameworkNonlocality}
\bibinfo{author}{\bibfnamefont{R.}~\bibnamefont{Gallego}},
  \bibinfo{author}{\bibfnamefont{L.~E.} \bibnamefont{W\"urflinger}},
  \bibinfo{author}{\bibfnamefont{A.}~\bibnamefont{Ac\'{\i}n}},
  \bibnamefont{and}
  \bibinfo{author}{\bibfnamefont{M.}~\bibnamefont{Navascu\'es}},
  \bibinfo{journal}{Physical Review Letters} \textbf{\bibinfo{volume}{109}},
  \bibinfo{pages}{070401} (\bibinfo{year}{2012}).

\bibitem[{\citenamefont{{d}e Vicente}(2014)}]{DeVincenteNL}
\bibinfo{author}{\bibfnamefont{J.~I.} \bibnamefont{{d}e Vicente}},
  \bibinfo{journal}{J. Phys. A: Math. Theor.} \textbf{\bibinfo{volume}{47}},
  \bibinfo{pages}{424017} (\bibinfo{year}{2014}).

\bibitem[{\citenamefont{Clauser et~al.}(1969)\citenamefont{Clauser, Horne,
  Shimony, and Holt}}]{CHSH}
\bibinfo{author}{\bibfnamefont{J.~F.} \bibnamefont{Clauser}},
  \bibinfo{author}{\bibfnamefont{M.~A.} \bibnamefont{Horne}},
  \bibinfo{author}{\bibfnamefont{A.}~\bibnamefont{Shimony}}, \bibnamefont{and}
  \bibinfo{author}{\bibfnamefont{R.~A.} \bibnamefont{Holt}},
  \bibinfo{journal}{Phys. Rev. Lett.} \textbf{\bibinfo{volume}{23}},
  \bibinfo{pages}{880} (\bibinfo{year}{1969}).

\bibitem[{\citenamefont{Elitzur et~al.}(1992)\citenamefont{Elitzur, Popescu,
  and
  Rohrlich}}]{ElitzurPopescuRohrlich1992:QuantumNonlocalityForEachPairInAnEnsemble}
\bibinfo{author}{\bibfnamefont{A.~C.} \bibnamefont{Elitzur}},
  \bibinfo{author}{\bibfnamefont{S.}~\bibnamefont{Popescu}}, \bibnamefont{and}
  \bibinfo{author}{\bibfnamefont{D.}~\bibnamefont{Rohrlich}},
  \bibinfo{journal}{Phys. Lett. A} \textbf{\bibinfo{volume}{162}},
  \bibinfo{pages}{25} (\bibinfo{year}{1992}).

\bibitem[{\citenamefont{Dantzig and
  Thapa}(2003)}]{DantzigThapa:LinearProgramming2}
\bibinfo{author}{\bibfnamefont{G.~B.} \bibnamefont{Dantzig}} \bibnamefont{and}
  \bibinfo{author}{\bibfnamefont{M.~N.} \bibnamefont{Thapa}},
  \emph{\bibinfo{title}{Linear programming 2}}, Springer Series in Operations
  Research and Financial Engineering (\bibinfo{publisher}{Springer Verlag},
  \bibinfo{year}{2003}).

\bibitem[{\citenamefont{Abramsky et~al.}(2016)\citenamefont{Abramsky,
  Constantin, and Ying}}]{AbramskyConstantinYing2015:HardyIsAlmostEverywhere}
\bibinfo{author}{\bibfnamefont{S.}~\bibnamefont{Abramsky}},
  \bibinfo{author}{\bibfnamefont{C.~M.} \bibnamefont{Constantin}},
  \bibnamefont{and} \bibinfo{author}{\bibfnamefont{S.}~\bibnamefont{Ying}},
  \bibinfo{journal}{Information and Computation}
  \textbf{\bibinfo{volume}{250}}, \bibinfo{pages}{3} (\bibinfo{year}{2016}).

\bibitem[{\citenamefont{Abramsky and
  Constantin}(2014)}]{abramsky2014classification}
\bibinfo{author}{\bibfnamefont{S.}~\bibnamefont{Abramsky}} \bibnamefont{and}
  \bibinfo{author}{\bibfnamefont{C.}~\bibnamefont{Constantin}},
  \bibinfo{journal}{Electronic Proceedings in Theoretical Computer Science}
  \textbf{\bibinfo{volume}{171}}, \bibinfo{pages}{10} (\bibinfo{year}{2014}).

\bibitem[{\citenamefont{Mansfield}(2013)}]{Mansfield:DPhil-thesis}
\bibinfo{author}{\bibfnamefont{S.}~\bibnamefont{Mansfield}},
  \bibinfo{type}{{DPhil} thesis}, \bibinfo{school}{University of Oxford}
  (\bibinfo{year}{2013}).

\bibitem[{\citenamefont{Abramsky and
  Hardy}(2012)}]{AbramskyHardy:LogicalBellIneqs}
\bibinfo{author}{\bibfnamefont{S.}~\bibnamefont{Abramsky}} \bibnamefont{and}
  \bibinfo{author}{\bibfnamefont{L.}~\bibnamefont{Hardy}},
  \bibinfo{journal}{Physical Review A} \textbf{\bibinfo{volume}{85}},
  \bibinfo{pages}{062114} (\bibinfo{year}{2012}).

\bibitem[{\citenamefont{Dzhafarov et~al.}(2015)\citenamefont{Dzhafarov, Kujala,
  and Cervantes}}]{dzhafarov2015contextuality}
\bibinfo{author}{\bibfnamefont{E.~N.} \bibnamefont{Dzhafarov}},
  \bibinfo{author}{\bibfnamefont{J.~V.} \bibnamefont{Kujala}},
  \bibnamefont{and} \bibinfo{author}{\bibfnamefont{V.~H.}
  \bibnamefont{Cervantes}}, in \emph{\bibinfo{booktitle}{International
  Symposium on Quantum Interaction}} (\bibinfo{organization}{Springer},
  \bibinfo{year}{2015}), pp. \bibinfo{pages}{12--23}.

\bibitem[{\citenamefont{Kaszlikowski et~al.}(2000)\citenamefont{Kaszlikowski,
  Gnaci{\'n}ski, {\.Z}ukowski, Miklaszewski, and
  Zeilinger}}]{kaszlikowski2000violations}
\bibinfo{author}{\bibfnamefont{D.}~\bibnamefont{Kaszlikowski}},
  \bibinfo{author}{\bibfnamefont{P.}~\bibnamefont{Gnaci{\'n}ski}},
  \bibinfo{author}{\bibfnamefont{M.}~\bibnamefont{{\.Z}ukowski}},
  \bibinfo{author}{\bibfnamefont{W.}~\bibnamefont{Miklaszewski}},
  \bibnamefont{and}
  \bibinfo{author}{\bibfnamefont{A.}~\bibnamefont{Zeilinger}},
  \bibinfo{journal}{Physical Review Letters} \textbf{\bibinfo{volume}{85}},
  \bibinfo{pages}{4418} (\bibinfo{year}{2000}).

\bibitem[{\citenamefont{Ac{\'i}n et~al.}(2002)\citenamefont{Ac{\'i}n, Durt,
  Gisin, and Latorre}}]{acin2002quantum}
\bibinfo{author}{\bibfnamefont{A.}~\bibnamefont{Ac{\'i}n}},
  \bibinfo{author}{\bibfnamefont{T.}~\bibnamefont{Durt}},
  \bibinfo{author}{\bibfnamefont{N.}~\bibnamefont{Gisin}}, \bibnamefont{and}
  \bibinfo{author}{\bibfnamefont{J.~I.} \bibnamefont{Latorre}},
  \bibinfo{journal}{Physical Review A} \textbf{\bibinfo{volume}{65}},
  \bibinfo{pages}{052325} (\bibinfo{year}{2002}).

\bibitem[{\citenamefont{Junge et~al.}(2010)\citenamefont{Junge, Palazuelos,
  P{\'e}rez-Garc{\'\i}a, Villanueva, and Wolf}}]{junge2010operator}
\bibinfo{author}{\bibfnamefont{M.}~\bibnamefont{Junge}},
  \bibinfo{author}{\bibfnamefont{C.}~\bibnamefont{Palazuelos}},
  \bibinfo{author}{\bibfnamefont{D.}~\bibnamefont{P{\'e}rez-Garc{\'\i}a}},
  \bibinfo{author}{\bibfnamefont{I.}~\bibnamefont{Villanueva}},
  \bibnamefont{and} \bibinfo{author}{\bibfnamefont{M.~M.} \bibnamefont{Wolf}},
  \bibinfo{journal}{Physical review letters} \textbf{\bibinfo{volume}{104}},
  \bibinfo{pages}{170405} (\bibinfo{year}{2010}).

\bibitem[{\citenamefont{Massar et~al.}(2002)\citenamefont{Massar, Pironio,
  Roland, and Gisin}}]{massar2002bell}
\bibinfo{author}{\bibfnamefont{S.}~\bibnamefont{Massar}},
  \bibinfo{author}{\bibfnamefont{S.}~\bibnamefont{Pironio}},
  \bibinfo{author}{\bibfnamefont{J.}~\bibnamefont{Roland}}, \bibnamefont{and}
  \bibinfo{author}{\bibfnamefont{B.}~\bibnamefont{Gisin}},
  \bibinfo{journal}{Physical Review A} \textbf{\bibinfo{volume}{66}},
  \bibinfo{pages}{052112} (\bibinfo{year}{2002}).

\bibitem[{\citenamefont{Laplante et~al.}(2012)\citenamefont{Laplante, Lerays,
  and Roland}}]{laplante2012classical}
\bibinfo{author}{\bibfnamefont{S.}~\bibnamefont{Laplante}},
  \bibinfo{author}{\bibfnamefont{V.}~\bibnamefont{Lerays}}, \bibnamefont{and}
  \bibinfo{author}{\bibfnamefont{J.}~\bibnamefont{Roland}}, in
  \emph{\bibinfo{booktitle}{International Colloquium on Automata, Languages,
  and Programming}} (\bibinfo{organization}{Springer}, \bibinfo{year}{2012}),
  pp. \bibinfo{pages}{617--628}.

\bibitem[{\citenamefont{Degorre et~al.}(2009)\citenamefont{Degorre, Kaplan,
  Laplante, and Roland}}]{degorre2009communication}
\bibinfo{author}{\bibfnamefont{J.}~\bibnamefont{Degorre}},
  \bibinfo{author}{\bibfnamefont{M.}~\bibnamefont{Kaplan}},
  \bibinfo{author}{\bibfnamefont{S.}~\bibnamefont{Laplante}}, \bibnamefont{and}
  \bibinfo{author}{\bibfnamefont{J.}~\bibnamefont{Roland}}, in
  \emph{\bibinfo{booktitle}{International Symposium on Mathematical Foundations
  of Computer Science}} (\bibinfo{organization}{Springer},
  \bibinfo{year}{2009}), pp. \bibinfo{pages}{270--281}.

\bibitem[{\citenamefont{Laplante et~al.}(2016)\citenamefont{Laplante,
  Lauri{\`e}re, Nolin, Roland, and Senno}}]{laplante2016robust}
\bibinfo{author}{\bibfnamefont{S.}~\bibnamefont{Laplante}},
  \bibinfo{author}{\bibfnamefont{M.}~\bibnamefont{Lauri{\`e}re}},
  \bibinfo{author}{\bibfnamefont{A.}~\bibnamefont{Nolin}},
  \bibinfo{author}{\bibfnamefont{J.}~\bibnamefont{Roland}}, \bibnamefont{and}
  \bibinfo{author}{\bibfnamefont{G.}~\bibnamefont{Senno}}, in
  \emph{\bibinfo{booktitle}{11th Conference on the Theory of Quantum
  Computation, Communication and Cryptography}} (\bibinfo{year}{2016}).

\bibitem[{\citenamefont{Barrett et~al.}(2006)\citenamefont{Barrett, Kent, and
  Pironio}}]{BarretKentPironio2006:MaximallyNonlocalAndMonogamousQuantumCorrelations}
\bibinfo{author}{\bibfnamefont{J.}~\bibnamefont{Barrett}},
  \bibinfo{author}{\bibfnamefont{A.}~\bibnamefont{Kent}}, \bibnamefont{and}
  \bibinfo{author}{\bibfnamefont{S.}~\bibnamefont{Pironio}},
  \bibinfo{journal}{Phys. Rev. Lett.} \textbf{\bibinfo{volume}{97}},
  \bibinfo{pages}{170409} (\bibinfo{year}{2006}).

\bibitem[{\citenamefont{Aolita et~al.}(2012)\citenamefont{Aolita, Gallego,
  Ac{\'\i}n, Chiuri, Vallone, Mataloni, and
  Cabello}}]{AolitaEtAl2012:FullyNonlocalQuantumCorrelations}
\bibinfo{author}{\bibfnamefont{L.}~\bibnamefont{Aolita}},
  \bibinfo{author}{\bibfnamefont{R.}~\bibnamefont{Gallego}},
  \bibinfo{author}{\bibfnamefont{A.}~\bibnamefont{Ac{\'\i}n}},
  \bibinfo{author}{\bibfnamefont{A.}~\bibnamefont{Chiuri}},
  \bibinfo{author}{\bibfnamefont{G.}~\bibnamefont{Vallone}},
  \bibinfo{author}{\bibfnamefont{P.}~\bibnamefont{Mataloni}}, \bibnamefont{and}
  \bibinfo{author}{\bibfnamefont{A.}~\bibnamefont{Cabello}},
  \bibinfo{journal}{Phys. Rev. A} \textbf{\bibinfo{volume}{85}},
  \bibinfo{pages}{032107} (\bibinfo{year}{2012}).

\bibitem[{\citenamefont{Greenberger et~al.}(1989)\citenamefont{Greenberger,
  Horne, and Zeilinger}}]{GHZ}
\bibinfo{author}{\bibfnamefont{D.~M.} \bibnamefont{Greenberger}},
  \bibinfo{author}{\bibfnamefont{M.~A.} \bibnamefont{Horne}}, \bibnamefont{and}
  \bibinfo{author}{\bibfnamefont{A.}~\bibnamefont{Zeilinger}}, in
  \emph{\bibinfo{booktitle}{{B}ell's theorem, quantum theory, and conceptions
  of the universe}}, edited by
  \bibinfo{editor}{\bibfnamefont{M.}~\bibnamefont{Kafatos}}
  (\bibinfo{publisher}{Kluwer}, \bibinfo{year}{1989}), pp.
  \bibinfo{pages}{69--72}.

\bibitem[{\citenamefont{Mermin}(1990)}]{Mermin1990:ExtremeQuantumEntanglementInASuperpositionOfMacroscopicallyDistinctStates}
\bibinfo{author}{\bibfnamefont{N.~D.} \bibnamefont{Mermin}},
  \bibinfo{journal}{Physical Review Letters} \textbf{\bibinfo{volume}{65}},
  \bibinfo{pages}{1838} (\bibinfo{year}{1990}).

\bibitem[{\citenamefont{Abramsky et~al.}(2012)\citenamefont{Abramsky,
  Mansfield, and Barbosa}}]{AbramskyMansfieldBarbosa:Cohomology-QPL}
\bibinfo{author}{\bibfnamefont{S.}~\bibnamefont{Abramsky}},
  \bibinfo{author}{\bibfnamefont{S.}~\bibnamefont{Mansfield}},
  \bibnamefont{and} \bibinfo{author}{\bibfnamefont{R.~S.}
  \bibnamefont{Barbosa}}, in \emph{\bibinfo{booktitle}{Proceedings 8th
  International Workshop on Quantum Physics and Logic, Nijmegen, Netherlands,
  October 27--29, 2011}}, edited by
  \bibinfo{editor}{\bibfnamefont{B.}~\bibnamefont{Jacobs}},
  \bibinfo{editor}{\bibfnamefont{P.}~\bibnamefont{Selinger}}, \bibnamefont{and}
  \bibinfo{editor}{\bibfnamefont{B.}~\bibnamefont{Spitters}}
  (\bibinfo{publisher}{Open Publishing Association}, \bibinfo{year}{2012}),
  vol.~\bibinfo{volume}{95} of \emph{\bibinfo{series}{Electronic Proceedings in
  Theoretical Computer Science}}, pp. \bibinfo{pages}{1--14}.

\bibitem[{\citenamefont{Cleve and Mittal}(2014)}]{CleveMittal2014}
\bibinfo{author}{\bibfnamefont{R.}~\bibnamefont{Cleve}} \bibnamefont{and}
  \bibinfo{author}{\bibfnamefont{R.}~\bibnamefont{Mittal}}, in
  \emph{\bibinfo{booktitle}{Automata, Languages, and Programming, ICALP 2014}},
  edited by \bibinfo{editor}{\bibfnamefont{J.}~\bibnamefont{Esparza}},
  \bibinfo{editor}{\bibfnamefont{P.}~\bibnamefont{Fraigniaud}},
  \bibinfo{editor}{\bibfnamefont{T.}~\bibnamefont{Husfeldt}}, \bibnamefont{and}
  \bibinfo{editor}{\bibfnamefont{E.}~\bibnamefont{Koutsoupias}}
  (\bibinfo{organization}{Springer}, \bibinfo{year}{2014}), vol.
  \bibinfo{volume}{8572} of \emph{\bibinfo{series}{Lecture Notes in Computer
  Science}}, pp. \bibinfo{pages}{320--331}.

\bibitem[{\citenamefont{Cleve et~al.}(2017)\citenamefont{Cleve, Liu, and
  Slofstra}}]{CleveLiuSlofstra17}
\bibinfo{author}{\bibfnamefont{R.}~\bibnamefont{Cleve}},
  \bibinfo{author}{\bibfnamefont{L.}~\bibnamefont{Liu}}, \bibnamefont{and}
  \bibinfo{author}{\bibfnamefont{W.}~\bibnamefont{Slofstra}},
  \bibinfo{journal}{Journal of Mathematical Physics}
  \textbf{\bibinfo{volume}{58}}, \bibinfo{pages}{012202}
  (\bibinfo{year}{2017}).

\bibitem[{\citenamefont{Man{\v{c}}inska and
  Roberson}(2016)}]{MancinskaRoberson2016}
\bibinfo{author}{\bibfnamefont{L.}~\bibnamefont{Man{\v{c}}inska}}
  \bibnamefont{and} \bibinfo{author}{\bibfnamefont{D.~E.}
  \bibnamefont{Roberson}}, \bibinfo{journal}{Journal of Combinatorial Theory,
  Series B} \textbf{\bibinfo{volume}{118}}, \bibinfo{pages}{228}
  (\bibinfo{year}{2016}).

\bibitem[{\citenamefont{Abramsky et~al.}(2017)\citenamefont{Abramsky, Barbosa,
  de~Silva, and Zapata}}]{AbramskyEtAl2017:QuantumMonad}
\bibinfo{author}{\bibfnamefont{S.}~\bibnamefont{Abramsky}},
  \bibinfo{author}{\bibfnamefont{R.~S.} \bibnamefont{Barbosa}},
  \bibinfo{author}{\bibfnamefont{N.}~\bibnamefont{de~Silva}}, \bibnamefont{and}
  \bibinfo{author}{\bibfnamefont{O.}~\bibnamefont{Zapata}},
  \emph{\bibinfo{title}{The quantum monad on relational structures}},
  \bibinfo{howpublished}{Submitted for publication} (\bibinfo{year}{2017}).

\end{thebibliography}
\fi

\ifbothtexts
\setcounter{affil}{0}
\newpage~\newpage
\fi

\setcounter{proposition}{3}

\ifsupplemental
\title{The Contextual Fraction as a Measure of Contextuality---Supplemental Material}
\reusableauthor{Samson Abramsky}%
\reusableaffiliation{\affilone}%
\reusableauthor{Rui Soares Barbosa}%
\reusableaffiliation{\affilone}%
\reusableauthor{Shane Mansfield}%
\reusableaffiliation{\affiltwo}%
\date{\today}

\reusablemaketitle

\appendix


\setcounter{equation}{5}
\setcounter{table}{1}
\setcounter{figure}{2}

\section{A. Contextual fraction and violations of Bell inequalities (Proof of Theorem~\ref{thm:CFBellineqs})}

Theorem~\ref{thm:CFBellineqs} establishes the close link between the contextual fraction of an empirical model and
the violation of generalised Bell inequalities by that model. In the main text, we presented this result along with
the broad idea of its proof,
with particular emphasis on the linear program that calculates the Bell inequality for which a given model achieves a maximal violation.
For completeness, we include here the detailed proof of the three statements.

As a preliminary remark, we note that the Strong Duality theorem of Linear Programming does indeed apply here, since the set of feasible solutions is bounded and non-empty. Indeed, non-emptiness holds since the zero vector is feasible, while boundedness holds since the constraints imply that every component of a vector satisfying them must lie in the unit interval.

\fullformtrue
\thmCFBellineqsSUPP
\begin{proof}
\textbf{\ref{item:atmost}}
This follows from the decomposition of $e$
into a non-contextual and a strongly contextual models,
\[e = \NCF(e) e^{NC} + \CF(e) e^{SC} \Mdot\]
For any Bell inequality, determined by $\tuple{\av,R}$, with $\|\av\| > R$,
the left-hand side of the inequality for the model $e$ adds up to
\begin{calculation}
\av \cdot \ve
\just={by decomposition above and linearity of $(\av \cdot \blank)$}
\NCF(e) \, \av \cdot \veNC \;+\; \CF(e) \, \av \cdot \veSC
\just\leq{since $\av \cdot \veNC \leq R$    
      and $\av \cdot \veSC \leq\|\av\|$}   
\NCF(e) \, R \;+\; \CF(e) \, \|\av\|
\ejust=
\NCF(e) \, R \;+\; \CF(e) \, R \;+\; \CF(e)\,(\|\av\| - R)
\just={by $\NCF(e) + \CF(e) = 1$}
R \;+\; \CF(e)\, (\|\av\|-R) \Mdot
\end{calculation}%
Therefore, since $\|\av\| > R$, we have, as required,
\[\frac{\av \cdot \ve - R}{\|\av\|-R} \;\leq\; \CF(e) \Mdot\]

\textbf{\ref{item:ineq}}
Recall the LP \eqref{LP:quantifying} that yields the (non-)contextual fraction of a model $e$.
Its symmetric dual is the following:
\begin{equation*}
\begin{alignedat}{3}
&\textLP{Find }       \;\;&&  \yy \in \RR^m
\\
&\textLP{minimising }   \;\;&& \ve \cdot \yy
\\
&\textLP{subject to } \;\;&& \transpose{\IM} \, \yy \,\geq\, \vone
\\
&\textLP{and }        \;\;&& \yy \,\geq\, \vzero
& \Mdot
\end{alignedat}
\end{equation*}
Under the transformation of variables
\begin{equation}\label{eq:defafromy}
\av \;\defeq\; |\M|\inv \vone \,-\, \yy \Mcomma
\end{equation}
where $|\M|$ is the number of maximal contexts in the scenario,
the above LP can then be equivalently restated as follows:
\begin{equation}\label{LP:Bellineq}
\begin{alignedat}{3}
&\textLP{Find }       \;\;&&  \av \in \RR^m
\\
&\textLP{maximising}   \;\;&& \av \cdot \ve
\\
&\textLP{subject to } \;\;&& \transpose{\IM} \, \av \,\leq\, \vzero
\\
&\textLP{and }        \;\;&& \av \,\leq\, |\M|\inv\vone
& \Mcomma
\end{alignedat}
\end{equation}
with solutions (resp. optimal solutions) of one optimisation problem corresponding bijectively to those of the other via the transformation \eqref{eq:defafromy}.
To see that this is the case, note that:
\begin{calculation}  
\transpose{\IM} \, \av\; \leq\;  \vzero
\just\bimplies{by eq. \eqref{eq:defafromy}}
\transpose{\IM} \, (|\M|\inv \vone \, - \,\yy)\; \leq\;  \vzero
\just\bimplies{by linearity}
|\M|\inv \transpose{\IM} \, \vone  -  \, \transpose{\IM} \, \yy \; \leq\;  \vzero
\just\bimplies{$\transpose{\IM}$ has exactly $|\M|$ 1-entries in each row}
|\M|\inv|\M|\vone  - \, \transpose{\IM} \, \yy \; \leq\;  \vzero
\ejust\bimplies
\transpose{\IM} \, \yy \; \geq\;  \vone 
\ejust{\text{and}}
\av \; \leq\; |\M|\inv \vone
\just\bimplies{by eq. \eqref{eq:defafromy}}
 |\M|\inv \vone - \yy \; \leq\; |\M|\inv \vone
\ejust\bimplies
\yy \; \geq\; \vzero \Mcomma
\end{calculation}%
showing that the feasibility conditions of the LPs \eqref{LP:dualquantifying} and \eqref{LP:Bellineq} are equivalent, and moreover that
\begin{calculationNUM}\label{eq:ave=1-yve}
\av \cdot \ve \nonumber
\just={by eq. \eqref{eq:defafromy}} 
\left(|\M|\inv \vone  - \yy \right) \cdot \ve \nonumber
\just={by linearity}
|\M|\inv \vone \cdot \ve - \yy \cdot \ve \nonumber
\just={$\ve$ consists of $|M|$ probability distributions}
1 - \yy \cdot \ve \Mcomma
\end{calculationNUM}%
showing that $\yy \cdot \ve$ is minimised exactly when $\av \cdot \ve$ is maximised.

The idea is that the components of a solution vector $\av$ (indexed by local assignments $\tuple{C,s}$)
are to be taken as the coefficients of an inequality, with bound $R=0$.

We first show that any feasible solution of the LP determines a Bell inequality,
\ie an inequality that is satisfied by all non-contextual models.
It suffices to show that it is satisfied by the deterministic non-contextual models---that is, those determined by a single global assignment $\fdec{g}{X}{O}$---since all other non-contextual models are convex combinations of these.
Recall that the columns of $\IM$ (and so the rows of $\transpose{\IM}$) are exactly the vectors representing these models.
Hence, the fact that $\av$ determines a Bell inequality is concisely expressed by the system of linear inequalities
\[
 \transpose{\IM} \, \av \;\leq\;  \vzero \Mcomma
\]
which is one of the feasibility conditions of our LP \eqref{LP:Bellineq}.


The other feasibility condition, $ \av \,\leq\, |\M|\inv \vone$,
is a bound on the components of $\av$.
This acts as a normalisation condition guaranteeing that,
for any feasible solution, the algebraic bound of the inequality (\ie its maximal violation) is at most $1$:
\begin{align*}
\|\av\| &\;=\; \sum_{C \in \M} \max \setdef{\av[\tuple{C,s}]}{s \in O^C}
\\ &\;\leq\; \sum_{C \in \M} |\M|\inv \\ & \;=\; 1 \Mdot
\end{align*}
Consequently,
\[
\max\{0,\av \cdot \ve\} 
\;\;\geq\;\;
(\av \cdot \ve) \; \max\{0,\|\av\|\} \Mcomma\]
and so (whenever the inequality corresponding to $\av$ is of any interest, \ie whenever $0 < \|\av\|$)
the normalised violation of the inequality by $e$ is at least $\av \cdot \ve$, the objective function that the LP maximises.
(In fact, for an optimal solution, the two bounds above are attained: provided $e$ is contextual, the algebraic bound of the optimal inequality is $1$ and its normalised violation by $e$ is simply given by $\av \cdot \ve$, as we shall see at the end of this item.)

We now show that an optimal solution, $\av^*$, to the LP yields our desired inequality, whose violation by $e$ is $\CF(e)$.
Let $\bb^*$ denote an optimal solution to the primal LP \eqref{LP:quantifying}, meaning that $\vone \cdot \bb^* = \NCF(e)$.
The Strong Duality theorem of Linear Programming (see e.g. \cite{DantzigThapa:LinearProgramming2})
then says that the dual LP
\eqref{LP:dualquantifying}
also admits an optimal solution $\yy^*$
and moreover that these two optimal solutions are related by
\begin{equation}\label{eq:LPstrongduality-yc}
\yy^* \cdot \ve\; = \; \vone \cdot \bb^* \Mdot
\end{equation}
Hence, writing $\av^*$ for the corresponding (via eq. \eqref{eq:defafromy})
optimal solution of \eqref{LP:Bellineq}, we have
\begin{calculation}
\av^* \cdot \ve
\just={by eq. \eqref{eq:ave=1-yve}}
1 - \yy^* \cdot \ve
\just={by eq. \eqref{eq:LPstrongduality-yc} (strong duality)}
1 - \vone \cdot \bb^*
\just={$\bb^*$ is optimal solution to LP \eqref{LP:quantifying}}
1 - \NCF(e)
\ejust=
\CF(e)
\end{calculation}%
This shows that the normalised violation of the inequality by the model $e$ is at least $\CF(e)$.
Since the opposite inequality follows from the first item,
this concludes the proof that the model $e$ attains a normalised violation of $\CF(e)$ of the Bell inequality
with coefficients $\av^*$ and bound $0$.
Incidentally, this implies in particular that the algebraic bound of this inequality, $\|\av^*\|$, is equal to $1$.

\textbf{\ref{item:tight}}
Consider any decomposition of the model $e$ (which satisfies $\CF(e)>0$) into a non-contextual and a strongly contextual parts,
with maximal possible weight on the former:
\begin{equation}\label{eq:pfitem3-decomp}
e \;=\; \NCF(e)\, e^{NC} \,+\, \CF(e) \, e^{SC} \Mdot
\end{equation}
We want to show that the non-contextual part of the model, $e^{NC}$, saturates the inequality from the previous item,
and that the strongly contextual part, $e^{SC}$, maximally violates it (\ie achieves the algebraic bound, which we know from the previous item to be $1$).
That is, the goal is to show that
\[\av^* \cdot \veNC = 0 \Mand \av^* \cdot \veSC = 1 \Mdot\]
Note that the inequalities
\begin{equation}\label{eq:av*veNCveSC-easyineq}
\av^* \cdot \veNC \leq 0 \Mand \av^* \cdot \veSC \leq 1 
\end{equation} 
follow from the fact that this is a Bell inequality for the bound $0$ and with algebraic bound $1$.

For the opposite inequalities, recall from the previous item that $\av^* \cdot \ve = \CF(e)$.
Thus, we have:
\begin{calculation}
\av^* \cdot \ve = \CF(e)
\just\bimplies{decomposition of eq. \eqref{eq:pfitem3-decomp} and linearity}
\NCF(e)\av^*\cdot\veNC + \CF(e)\av^* \cdot \veSC = \CF(e)
\ejust\bimplies
\NCF(e)\av^*\cdot\veNC + \CF(e)(\av^* \cdot \veSC - 1) = 0
\end{calculation}%
But the left-hand side now is a convex combination of two numbers which we know to be non-positive from eq. \eqref{eq:av*veNCveSC-easyineq}. This can only be equal to zero when we have
\[\NCF(e)\av^*\cdot\veNC = 0 \Mand \CF(e)(\av^* \cdot \veSC - 1) = 0\]
hence we have that
\[\av^* \cdot \veNC = 0 \Mand \av^* \cdot \veSC = 1 \]
as long as $\NCF(e) \neq 0$ and $\CF(e) \neq 0$, respectively.
\end{proof}

\begin{remark}
Incidentally, note that no use was made of the assumption that $\CF(e) > 0$ in the proof of item \ref{item:ineq}, and similarly for the first part of item \ref{item:tight}.
However, when $\CF(e)=0$, \ie when the model $e$ is non-contextual, it may happen that the Bell inequality obtained by the method described is trivial, in the sense that it is satisfied not only by all non-contextual models but also by all no-signalling ones.
In that situation, it is an inequality (in fact, an equality) defining the affine subspace of vectors corresponding to no-signalling models.
Indeed, this necessarily happens for models in the relative interior of the non-contextual polytope, for such a model cannot saturate any proper Bell inequality that separates the non-contextual from general no-signalling empirical models.
Therefore, even though the statement would strictly remain true given our definitions,
we have chosen to include the extra assumption to avoid it being misconstrued.
\end{remark}

As noted in the main text, decompositions of the form $e = \NCF(e)e^{NC} + \CF(e)e^{SC}$ are not necessarily unique.
This can happen when there is a face of the no-signalling polytope
consisting only of strongly contextual models (\ie whose vertices are all strongly contextual)
that is parallel to a face of the non-contextual polytope.
If these faces have dimension at least $1$ and the model $e$ lies in between them,
then any line going through $e$ and intersecting the two faces determines two models $e^{NC}$ and $e^{SC}$, corresponding to the intersections. For all these lines, the value of $\lambda$ in the decomposition 
$e = \lambda e^{NC} + (1-\lambda)e^{SC}$ will be the same.
An example of non-uniqueness in the $(3,2,2)$ Bell scenario is given by the models in Table~\ref{tab:non-unique-decomposition}.
On the other hand,
non-uniqueness cannot arise when the scenario is such that there are no two adjacent strongly contextual vertices of the no-signalling polytope, and hence no face of the polytope consisting solely of strongly contextual models: this is the case, for example, for the $(2,2,2)$ Bell scenario.


\begin{table*}
\caption{\label{tab:non-unique-decomposition}Empirical models $e^{SC}_1$, $e^{NC}_1$, $e^{SC}_2$, and $e^{NC}_2$ in a (3,2,2) Bell scenario, illustrating non-uniqueness of decomposition \eqref{eq:decomp}. The models on the left are strongly contextual, those on the right are non-contextual, and $\nfrac{1}{2}e^{NC}_1 + \nfrac{1}{2}e^{SC}_1 = \nfrac{1}{2}e^{NC}_2 + \nfrac{1}{2}e^{SC}_2$ is an empirical model with contextual fraction $\nfrac{1}{2}$.}
\begin{tabular}{ccc}
$e^{SC}_1$ & & $e^{NC}_1$
\\
$\ABCTableGHZ$ & \quad\quad\quad &  $\ABCTableLocalGHZ$
\\~\\~\\
$e^{SC}_2$ & & $e^{NC}_2$
\\
$\ABCTableOtherStronglyNL$ & & $\ABCTableOtherLocal$
\end{tabular}
\end{table*}

\section{B. Computational explorations}\label{appendix:computationalexplorations}

Computational tools in the form of a \textit{Mathematica} package have been implemented, which can compute:
\begin{enumerate}
\item
the empirical model arising from any quantum state and any sets of compatible measurements;
\item
the incidence matrix for any measurement scenario;
\item
the contextual fraction of any empirical model using LP \eqref{LP:quantifying}; 
\item
the Bell inequality of Theorem~\ref{thm:CFBellineqs},
using the dual LP \eqref{LP:dualquantifying} (under change of variables).
\end{enumerate}
We stress that these tools are completely general:
they can be applied to any pure or mixed quantum state in any Hilbert space and to any sets of compatible observables in that space,
including Bell scenarios as a special case.

\prlsection{Equatorial measurements on the Bell state $\BellState$}%
As an example of how the package can be used, we consider a family of empirical models that can be obtained by considering local measurements on the two-qubit Bell state
\[\BellState = \BellStateExpand\]
Recall that projective measurements on a qubit can equivalently be represented by a point on the Bloch sphere. Suppose that we allow the same two local measurements on each qubit, and that these are equatorial on the Bloch sphere, parametrised by angles $\phi_1$ and $\phi_2$ as in Figure \ref{fig:bloch}.
We assume these angles are in the interval $[0,\pi)$ since $\phi$ and $\phi+\pi$ correspond to the same measurement up to relabelling the outcomes.
One such model is the Bell--CHSH model from Table \ref{tab:BellPRmodel}, which is obtained when
\[(\phi_1,\phi_2) =  (0,\pi/3) \Mdot\]

\begin{figure}
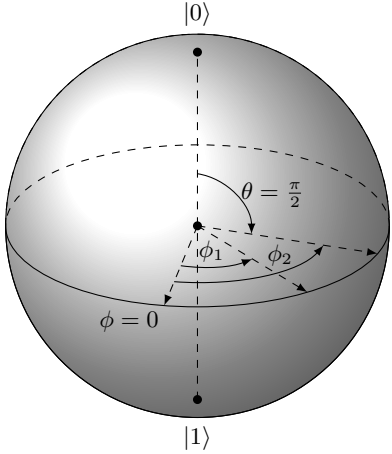

\EquatorialBloch
\caption{\label{fig:bloch} Equatorial measurements at $\phi_1$ and $\phi_2$ on the Bloch sphere.}
\end{figure}

We can use the package to plot the non-contextual fraction
of the resulting models as a function of $\phi_1$ and $\phi_2$ (Figure \ref{fig:plots}--(a)).
It is interesting to note that the Bell--CHSH model from Table \ref{tab:BellPRmodel}
does not achieve the maximum possible degree of contextuality among these models.
Instead, the maxima of the plot occur when
\[
\left\{\phi_1, \phi_2\right\} \in \left\{ \left\{ \frac{\pi}{8},\frac{5 \pi}{8} \right\}, \left\{ \frac{7 \pi}{8},\frac{3 \pi}{8} \right\} \right\} \Mdot
\]
All of the corresponding empirical models take the form of Table~\ref{tab:comp}, with
\[
p = \frac{ \sqrt{2} + 2}{8}.
\]
These can easily be seen to achieve the Tsirelson violation of the CHSH inequality.
Note that none of these models are strongly contextual: this observation provided one motivation
to look for proofs that Bell states cannot witness logical forms of contextuality with a finite number of measurements \cite{Mansfield:DPhil-thesis,AbramskyConstantinYing2015:HardyIsAlmostEverywhere},
although they do so at the limit where the number of measurement settings tends to infinity \cite{BarretKentPironio2006:MaximallyNonlocalAndMonogamousQuantumCorrelations}.

\begin{table}
\caption{\label{tab:comp}Empirical models corresponding to maxima of the plot shown in Figure \ref{fig:plots}--(a), where $p = \nfrac{\sqrt{2} + 2}{8}$. These achieve the Tsirelson violation of the CHSH inequality.}
\[ \ABTableNUM
{         p        & (\nfrac{1}{2}-p) &  (\nfrac{1}{2}-p)  &         p        }
{ (\nfrac{1}{2}-p) &         p        &          p         & (\nfrac{1}{2}-p) }
{ (\nfrac{1}{2}-p) &         p        &          p         & (\nfrac{1}{2}-p) }
{ (\nfrac{1}{2}-p) &         p        &          p         & (\nfrac{1}{2}-p) }
\]
\end{table}


It may seem surprising at first that the degree of contextuality of the empirical models is not constant with respect to the relative angle $(\phi_2 - \phi_1)$ between measurements, a fact that is apparent from the plot. For example, the empirical model obtained when $(\phi_1, \phi_2)=(0, \pi/2)$ is local, but if these values are shifted by $\pi/8$ the resulting empirical model achieves the maximum violation of the CHSH inequality.
As it happens, a rotation by $\Phi$ around the $Z$-axis for each of the qubits is described by
\begin{multline}\label{eq:rotation}
\left( \begin{array}{cc} e^{-i \Phi /2} & 0 \\ 0 & e^{i \Phi / 2} \end{array} \right) \otimes \left( \begin{array}{cc} e^{-i \Phi /2} & 0 \\ 0 & e^{i \Phi / 2} \end{array} \right) \\ = \left( \begin{array}{cccc} e^{-i \Phi} & 0 & 0 & 0 \\ 0 & 1 & 0 & 0 \\ 0 & 0 & 1 & 0 \\ 0 & 0 & 0 & e^{i \Phi} \end{array} \right) \Mdot
\end{multline}
If one equivalently thinks of leaving the measurements fixed and applying the rotations to the state instead,
note that this indeed introduces a relative phase of $2 \Phi$ between the terms in $\BellState$,
explaining the difference in the resulting empirical models.

\prlsection{Equatorial measurements on $n$-partite GHZ states}%
We can consider similar families of models for the $n$-partite GHZ states \cite{GHZ}, given for each $n>2$ by:
\begin{equation}\label{eq:GHZn-def}
\GHZStateNN = \GHZStateNExpand
\end{equation}
Note that with $n=2$ this would simply reduce to the $\BellState$ Bell state.
For $n>2$, Mermin considered the situation where each each of the $n$ parties can perform Pauli $X$ or $Y$ measurements,
and gave logical proofs of strong contextuality (nonlocality) via a parity argument which he called `all versus nothing'  \cite{Mermin1990:ExtremeQuantumEntanglementInASuperpositionOfMacroscopicallyDistinctStates}.
In \cite{AbramskyMansfieldBarbosa:Cohomology-QPL} it was shown that a very general form of all-\textit{vs}-nothing contextuality implies strong contextuality.

As before, we choose any two equatorial measurements on the Bloch sphere, and make these same two measurements available at each qubit.
For $\GHZStateN{3}$ and $\GHZStateN{4}$ we obtain the plots shown in Figures \ref{fig:plots}--(b) and \ref{fig:plots}--(c), respectively.
The maxima of the plot for the tripartite state reach $\CF(e)=1$, indicating strong contextuality, and occur when
\begin{equation}\label{eq:ghz3min}
\left\{\phi_1, \phi_2\right\} \in \left\{ \left\{\frac{\pi}{2},0\right\}, \left\{\frac{2 \pi}{3},\frac{\pi}{6}\right\}, \left\{\frac{5 \pi}{6},\frac{\pi}{3}\right\} \right\} \Mdot
\end{equation}
Of course, $(\phi_1,\phi_2) = (\pi/2,0)$ correspond to the Pauli measurements $Y$ and $X$, respectively, yielding the usual GHZ--Mermin model.
The empirical models corresponding to other maxima are identical up to re-labelling,
and so these provide alternative sets of measurements that can be made on the GHZ state and still lead to the familiar all-\textit{vs}-nothing argument for contextuality.

The situation is similar for $n=4$, in which maxima of $\CF(e)=1$ are seen to occur at
\begin{equation}\label{eq:ghz4min}
\left\{\phi_1, \phi_2\right\} \in
\left\{ \left\{\frac{\pi}{2},0\right\}, \left\{\frac{5 \pi}{8},\frac{\pi}{8}\right\}, \left\{\frac{3 \pi}{4},\frac{\pi}{4}\right\},\left\{\frac{7 \pi}{8},\frac{3 \pi}{8}\right\} \right\} \Mdot
\end{equation}
We can see a pattern beginning to emerge in (\ref{eq:ghz3min}) and (\ref{eq:ghz4min}), which leads to the following proposition.

\begin{proposition}\label{prop:equatorial-GHZn}
Equatorial measurements at
\[
(\phi_1, \phi_2) \in \setdef{ \left(\frac{(n+k) \, \pi}{2 n},\frac{k \, \pi}{2 n}\right)}{0 \leq k < n}
\]
on each qubit of the $\GHZStateNN$ state
give rise (up to relabelling of measurements and outcomes) to the strongly contextual GHZ--Mermin $n$-partite empirical model.
\end{proposition}

\begin{proof}
First, we know that this holds for $k = 0$, since in that case we simply have local Pauli $X$ and $Y$ measurements, the measurements prescribed in Mermin's argument.
For $0<k<n$, we can apply at each qubit a rotation around the $Z$ axis by the phase $\Phi \defeq \nfrac{k \, \pi}{2n}$, so that we continue to deal with the $X$ and $Y$ measurements.
It is necessary, however, to take account of the relative phase introduced by these operations on the $n$-qubit state.
By generalising (\ref{eq:rotation}) it is clear that the state obtained from $\GHZStateNN$ after rotating each qubit by $\Phi$ is
\[
\ket{\text{GHZ}(n,\Phi)} = \frac{1}{\sqrt{2}} \left( \ket{0 \cdots 0} + e^{i n \Phi} \ket{1\cdots1} \right) \Mdot
\]
Notice that for the relevant values of $\Phi$ the relative phase is a multiple of $\nfrac{\pi}{2}$.
Rotating only one of the qubits by $\nfrac{\pi}{2}$ an appropriate number of times brings us back to the state $\GHZStateNN$,
while each step changes the measurements $X$ and $Y$ to $-Y$ and $X$, respectively. Since $-Y$ is just the measurement $Y$ with a relabelling of the outcomes,
the whole change simply amounts to relabelling of measurements and  their outcomes.
\end{proof}

\section{C. Monotonicity (Proof of Theorem~\ref{thm:CFmonotone})}

First, we formally  define the operations of translation of measurements and coarse-graining of outcomes.
For the former, given an empirical model $e$ on the measurement scenario $\XpMpO$, a second measurement scenario $\XMO$, and a function $\fdec{f}{X}{X'}$ that preserves contexts,
\ie such that $C \in \M$ implies
$f(C) \subseteq C'$ for some $C' \in \M'$,
we define the empirical model $f^*e$ in $\XMO$ by pulling $e$ back along the map $f$: for each $C \in \M$ and $s \in O^C$,
\[(f^*e)_C(s) \defeq \sum_{t\in O^{f(C)}, t \circ f|_{C} = s}e_{f(C)}(t) \Mdot\]
For the latter, given an empirical model $e$ on the measurement scenario $\XMOp$ and a function $\fdec{h}{O'}{O}$, we define an empirical model $e/f$ on the scenario $\XMO$
as follows: for each $C \in \M$ and $s \in O^C$ 
\[(e/h)_C(s) \defeq \sum_{s'\in O'^C, h \circ s' = s}e_C(s')\Mdot\]

Two of the three combining operations were already introduced formally in the main text.
As for mixing, it has the obvious definition: given empirical models $e$ and $e'$ in $\XMO$ and a weight $\lambda \in [0,1]$,
the empirical model $\lambda e + (1-\lambda) e'$ is defined, for each $C \in \M$ and $s \in O^C$, as
\[(\lambda e + (1-\lambda) e')_C(s) \defeq \lambda \, e_C(s) \; + \; (1-\lambda) \, e'_C(s) \Mdot\]

We illustrate the use of the operations to construct local models for Bell scenarios, as mentioned in the main text.
We start from the generator $G = \enset{ m \longmapsto *}$, which sends a single measurement, deterministically, to a single outcome. We define functions $\fdef{f_i}{*}{i}$, $i \in \enset{1,\ldots,l}$ to relabel outcomes of measurements into the set $O = \enset{ 1,\ldots, l}$.
A deterministic model for a single agent with $k$ measurements $\enset{m_1,\ldots, m_k}$,
where measurement $m_j$ is assigned outcome $i_j$, is described (up to isomorphic relabelling) by
$\&_{j=1}^k G/f_{i_j}$.
Let $A$ be one such model, and $B$ another. Then the corresponding bipartite local model is described by $A \otimes B$. This obviously generalizes to any number of parties. Finally, any local model can be expressed as a mixture of deterministic local models.

We now present a detailed proof of Theorem~\ref{thm:CFmonotone}
which states the monotonicity properties of the measure of contextuality $\CF$ with respect to these five operations.

\thmCFmonotoneSUPP
\begin{proof}
We shall prove the equivalent statements in terms of $\NCF$ instead of $\CF$.

\textbf{Translation of measurements.} 
Let $e$ be a model in $\XpMpO$ and  $\fdec{f}{X}{X'}$ context-preserving.
$\NCF(e)$ is the maximal weight of a subprobability distribution $b$ on global assignments $O^{X'}$ such that $b|_{C'} \leq e_{C'}$ for any $C' \in \M'$.
Let $b_e$ be such a probability distribution of maximal weight
\ie the corresponding vector $\bb^*$ is an optimal solution to the LP \eqref{LP:quantifying}.

Define $f^*b_e$ a subprobability distribution on $O^X$ by, for any $g \in O^X$,
\begin{equation}\label{eq:defb-meas}
 (f^*b_e)(g) \defeq \sum_{g' \in O^{X'}, g' \circ f = g} b_e(g') \Mdot
\end{equation}
Note that this $f^*b_e$ has the same weight as $b_e$ since each $g' \in O^{X'}$ contributes to a single $g \in O^X$.

For any $C \in \M$ and $s \in O^C$, we have
\begin{calculation}
f^*b_e|_C(s)
\just={definition of marginalisation}
\sum_{g \in O^X, g|_C = s} f^*b_e(g)
\just={definition of $f^*b_e$, eq. \eqref{eq:defb-meas}}
\sum_{g \in O^X, g|_C = s} \;\, \sum_{g' \in O^{X'}, g' \circ f = g} b_e(g')
\ejust=
\sum_{g' \in O^{X'},  (g' \circ f)|_C=s}b_e(g')
\ejust=
\sum_{g' \in O^{X'},  g'|_{f(C)} \circ f|_C = s}b_e(g')
\ejust=
\sum_{t \in O^{f(C)}, t \circ f|_C = s} \;\, \sum_{g' \in O^{X'}, g'|_{f(C)} = t}b_e(g')
\just={definition of marginalisation}
\sum_{t \in O^{f(C)}, t \circ f|_C = s}b_e|_{f(C)}(t)
\just\leq{$b_e|_{f(C)}$ is subdistribution of $e_{f(C)}$}
\sum_{t \in O^{f(C)}, t \circ f|_C = s}e_{f(C)}(t)
\just={definition of $f^*e$}
(f^*e)_C(s)
\end{calculation}%
\ie $f^*b_e$ corresponds to a feasible solution to the LP for model $f^*e$. Hence,
\[\NCF(f^*e) \geq \weight{f^*b_e}
= \weight{b_e} = \NCF(e) \Mdot\]

\textbf{Coarse-graining of outcomes.}
Let $e$ be a model in $\XMOp$ and $\fdec{h}{O'}{O}$.
Again, write $b_e$ for a subprobability distribution on ${O'}^X$ of maximal weight $\NCF(e)$
satisfying $b_e|_C \leq e_C$ for all $C \in \M$.

Define $b_e/h$ a subprobability distribution on $O^X$ by, for any $g \in O^X$,
\begin{equation}\label{eq:defc-outcomes}
(b_e/h)(g) \defeq \sum_{g' \in O'^X, h \circ g' = g}b_e(g')\Mdot
\end{equation}
Similarly, note that $b_e/h$ has the same weight as $b_e$ since each $g' \in O'^{X}$ contributes to a single $g \in O^X$.

For any $C \in \M$ and $s \in O^C$, we have
\begin{calculation}
(b_e/h)|_C(s)
\just={definition of marginalisation}
\sum_{g \in O^X, g|_C = s} (b_e/h)(g)
\just={definition of $b_e/h$, eq. \eqref{eq:defc-outcomes}}
\sum_{g \in O^X, g|_C = s} \;\, \sum_{g' \in O'^X, h \circ g' = g}b_e(g')
\ejust=
\sum_{g' \in O'^X, (h \circ g')|_C = s}b_e(g')
\ejust=
\sum_{g' \in O'^X, h \circ g'|_C = s}b_e(g')
\ejust=
\sum_{s'\in O'^C, h \circ s' = s} \;\, \sum_{g' \in O'^X, g'|_C = s'}b_e(g')
\just={definition of marginalisation}
\sum_{s'\in O'^C, h \circ s' = s}b_e|_C(s')
\just\leq{$b_e|_C$ is subdistribution of $e_C$}
\sum_{s'\in O'^C, h \circ s' = s}e_C(s')
\just={definition of $e/h$}
(e/h)_C(s)
\end{calculation}%
\ie $b_e/h$ corresponds to a feasible solution to the LP for model $e/h$.
Therefore,
\[\NCF(e/h) \geq \weight{b_e/h}
= \weight{b_e} = \NCF(e) \Mdot\]


\textbf{Mixing.}
Let $e_1$ and $e_2$ be models in $\XMO$ and $\lambda \in [0,1]$. 
Note that in terms of vector representation we have
$\veA{\lambda e_1 + (1-\lambda) e_2} = \lambda \veA{e_1} + (1-\lambda) \veA{e_2}$.

Let $\bb^*_i$ be an optimal solution to the LP \eqref{LP:quantifying} relative to model $e_i$ ($i \in \enset{1,2}$),
and set $\bb \defeq \lambda \bb^*_1 + (1-\lambda) \bb^*_2$.
Then, $\bb \geq 0$ follows from non-negativity of $\bb^*_1$ and $\bb^*_2$,
and moreover
\begin{calculation}
\IM \, \bb
\ejust=
\IM \, (\lambda \bb^*_1 + (1-\lambda) \bb^*_2)
\just={linearity}
\lambda \IM \, \bb^*_1 + (1-\lambda) \IM \, \bb^*_2
\just\leq{feasibility of $\bb^*_i$, $\IM \, \bb^*_i \leq \veA{e_i}$}
\lambda \veA{e_1}  + (1-\lambda) \veA{e_2}
\ejust=
\veA{\lambda e_1 + (1-\lambda) e_2} \Mdot
\end{calculation}%
This means that $\bb$ is a feasible solution to the LP relative to the model $\lambda e_1 + (1-\lambda)e_2$,
achieving the following value of the objective function:
\begin{calculation}
\vone \cdot \bb
\ejust=
\vone \cdot (\lambda \bb^*_1 + (1-\lambda) \bb^*_2)
\just={linearity}
\lambda (\vone \cdot \bb^*_1) + (1-\lambda)(\vone \cdot \bb^*_2)
\just={optimality of $\bb^*_i$, $\vone \cdot \bb^*_i = \NCF(e_i)$}
\lambda \, \NCF(e_1) \,+\, (1-\lambda) \, \NCF(e_2) \Mdot
\end{calculation}%
Since $\NCF(\lambda e_1 + (1-\lambda) e_2)$ is the optimal (maximum) value for this primal LP, we have that
\[\NCF(\lambda e_1 + (1-\lambda)e_2) \,\geq\, \lambda \,\NCF(e_1) \,+\, (1-\lambda) \, \NCF(e_2)\]
as desired.

\textbf{Product.}
Let $e_1$ and $e_2$ be models on $\XMOone$ and $\XMOtwo$, respectively.
Note that global assignments for the scenario $\tuple{X_1 \sqcup X_2, \M_1 \star \M_2,O}$
are in bijective correspondence with tuples $\tuple{g_1,g_2}$ where $\fdec{g_i}{X_i}{O}$ is a global assignment for $\XMOi$ ($i\in\enset{1,2}$).
Similarly, contexts $C \in \M_1 \star \M_2$ are those of the form $C = C_1 \sqcup C_2$ with 
$C_i \in \M_i$, hence they are in bijective correspondence with pairs $\tuple{C_1,C_2}$ of contexts with $C_i \in \M_i$. Consequently,
local assignments $\tuple{C_1 \sqcup C_2 \in \M_1\star\M_2, \fdec{s}{C_1 \sqcup C_2}{O}}$
of the scenario $\tuple{X_1 \sqcup X_2, \M_1 \star \M_2,O}$
are in bijective correspondence with pairs of local assignments for each of the scenarios,
\ie pairs
$\tuple{\tuple{C_1,s_1},\tuple{C_2,s_2}}$ with $\tuple{C_i \in \M_i,\fdec{s_i}{C_i}{O}}$.
We use these equivalent representations to index the representation
of empirical models as vectors, the incidence matrix, etc. for the scenario  $\tuple{X_1 \sqcup X_2, \M_1 \star \M_2,O}$.

Observe that, in terms of the vector representations,
the product empirical model $e_1 \otimes e_2$ is concisely written as
$\veA{e_1 \otimes e_2} = \veA{e_1} \otimes \veA{e_2}$,
since for any local assignments $\tuple{C_i \in \M_i, \fdec{s_i}{C_i}{O}}$,
\[\veA{e_1 \otimes e_2}[\tuple{s_1,s_2}] = \veA{e_1}[s_1] \veA{e_2}[s_2] \Mdot\]
Moreover, if $\IM_1$ and $\IM_2$ are the incidence matrices for each of the measurement scenarios,
then the incidence matrix $\IM$ for the scenario $\tuple{X_1 \sqcup X_2, \M_1 \star \M_2,O}$ is precisely
$\IM = \IM_1 \otimes \IM_2$
since, for global assignments $\fdec{g_i}{X_i}{O}$ and local assignments $\tuple{C_i \in \M_i, \fdec{s_i}{C_i}{O}}$,
\begin{align*}
& \IM[\tuple{g_1,g_2}, \tuple{\tuple{C_1,s_1},\tuple{C_2,s_2}}] 
\\ = \; &
\begin{cases}
 1 \;\text{ if $\tuple{g_1,g_2}|_{C_1 \sqcup C_2} = \tuple{s_1,s_2}$;}
\\ 
 0 \;\text{ otherwise.}
\end{cases}
\\ = \; &
\begin{cases}
 1 \;\text{ if $g_1|_{C_1} = s_1$ and $g_2|_{C_2} = s_2$;}
\\ 
 0 \;\text{ otherwise.}
\end{cases}
\\ = \; &
\IM_1[g_1,\tuple{C_1,s_1}] \, \IM_2[g_2,\tuple{C_2,s_2}] \Mdot
\end{align*}

Let $\bb^*_1$ and $\bb^*_2$ be optimal solutions to the primal LP \eqref{LP:quantifying} relative to $e_1$ and $e_2$, respectively,
and set $\bb \defeq \bb^*_1 \otimes \bb^*_2$. 
Then we have $\bb \geq \vzero$ from the non-negativity of the $\bb^*_i$, and moreover
\begin{calculation}
\IM \, \bb
\ejust=
(\IM_1 \otimes \IM_2)\,(\bb^*_1 \otimes \bb^*_2)
\ejust=
(\IM_1 \, \bb^*_1) \otimes (\IM_2 \, \bb^*_2)
\just\leq{by feasibility of $\bb^*_i$, $\IM_i\,\bb^*_i \leq \veA{e_i}$}
\veA{e_1} \otimes \veA{e_2}
\ejust=
\veA{e_1 \otimes e_2} \Mcomma
\end{calculation}%
hence $\bb$ is a feasible solution to the primal LP  relative to $e_1 \otimes e_2$,
achieving the following value of the objective function:
\begin{calculation}
\vone \cdot \bb
\ejust=
(\vone \otimes \vone) \cdot (\bb^*_1 \otimes \bb^*_2) 
\ejust=
(\vone \cdot \bb^*_1) (\vone \cdot \bb^*_2)
\just={by optimality of $\bb^*_i$, $\vone \cdot \bb^*_i = \NCF(e_1)$}
\NCF(e_1)\NCF(e_2) \Mdot
\end{calculation}%
Since $\NCF(e_1 \otimes e_2)$ is the optimal (maximum) value for this primal LP, we have that
\[\NCF(e_1 \otimes e_2) \geq \NCF(e_1)\NCF(e_2) \Mdot\]

For the opposite inequality, we follow an analogous argument using the dual LP \eqref{LP:dualquantifying}.
Let $\yy^*_1$ and $\yy^*_2$ be optimal solutions to the dual LP for $e_1$ and $e_2$, respectively,
and set $\yy \defeq \yy^*_1 \otimes \yy^*_2$.
Then we have $\yy \geq 0$ from the non-negativity of the $\yy^*_i$, and moreover
\begin{calculation}
\transpose{\IM} \, \yy
\ejust=
(\transpose{\IM_1} \otimes \transpose{\IM_2})\,(\yy^*_1 \otimes \yy^*_2)
\ejust=
(\transpose{\IM_1} \, \yy^*_1) \otimes (\transpose{\IM_2} \, \yy^*_2)
\just\geq{by feasibility of $\yy^*_i$, $\transpose{\IM_i}\,\yy^*_i \geq \vone$}
\vone \otimes \vone
\ejust=
\vone
\end{calculation}%
hence $\yy$ is a feasible solution to the dual LP relative to $e_1 \otimes e_2$,
achieving the following value of the objective function:
\begin{calculation}
\yy \cdot \veA{e_1 \otimes e_2}
\ejust=
(\yy^*_1 \otimes \yy^*_2) \cdot (\veA{e_1} \otimes \veA{e_2})
\ejust=
(\yy^*_1 \cdot \veA{e_1}) (\yy^*_2 \cdot \veA{e_2})
\just={by optimality of $\yy^*_i$, $\yy^*_i \cdot \veA{e_i} = \NCF(e_i)$}
\NCF(e_1)\NCF(e_2) \Mdot
\end{calculation}%
Since $\NCF(e_1 \otimes e_2)$ is the optimal (minimum) value for this dual LP, we have that
\[\NCF(e_1 \otimes e_2) \leq \NCF(e_1)\NCF(e_2) \Mdot\]

\textbf{Choice.}
Let $e_1$ and $e_2$ be models on $\XMOone$ and $\XMOtwo$, respectively,
and consider the model $e_1 \choice e_2$ on $\tuple{X_1 \sqcup X_2, \M_1 \sqcup \M_2, O}$.
Write $b_{e_1 \choice e_2}$ for a subprobability distribution on global assignments $O^{X_1 \sqcup X_2} \cong O^{X_1} \times O^{X_2}$
of maximal weight $\NCF(e_1 \choice e_2)$ that satisfies
$b_{e_1 \choice e_2}|_C \leq (e_1 \choice e_2)_C$ for any $C \in \M_1 \sqcup \M_2$.

For $i \in \enset{1,2}$, define a subprobability distribution $b_i$ on $O^{X_i}$ by $b_i \defeq b_{e_1 \choice e_2}|_{X_i}$.
Note that each $b_i$ has the same weight as $b_{e_1 \choice e_2}$.
Then, for any $C \in \M_i$ and $s \in O^C$, we have
\begin{calculation}
b_i|_C(s)
\just={definition of $b_i$}
(b_{e_1 \choice e_2}|_{X_i})|_C(s)
\ejust= 
b_{e_1 \choice e_2}|_C(s)
\just\leq{by feasibility of $b_{e_1 \choice e_2}$ for $e_1 \choice e_2$}
(e_1 \choice e_2)_C(s)
\just={definition of $e_1 \choice e_2$}
(e_i)_C(s)
\end{calculation}%
That is, $b_i$ corrresponds to a feasible solution of the LP for the model $e_i$,
implying that 
\[\NCF(e_i) \geq \weight{b_i} = \weight{b_{e_1 \choice e_2}} = \NCF(e_1 \choice e_2) \Mdot\]
Therefore, $\NCF(e_1 \choice e_2) \leq \min\{\NCF(e_1),\NCF(e_2)\}$.

For the opposite inequality,
let $b_{e_1}$ and $b_{e_2}$ be subprobability distributions on $O^{X_1}$ and $O^{X_2}$, respectively,
corresponding to optimal solutions to the LP for $e_1$ and $e_2$, respectively.
The goal is to define a subprobability distribution $b$ on $O^{X_1 \sqcup X_2} \cong O^{X_1} \times O^{X_2}$
with weight $\weight{b} = \min\{\weight{b_{e_1}},\weight{b_{e_2}}\}$
such that
\begin{equation}\label{eq:condition-cmarg}
b|_{X_1} \leq b_{e_1} \Mand b|_{X_2} \leq b_{e_2} \Mdot
\end{equation}
This condition guarantees that this is a feasible solution to the LP for $e_1 \choice e_2$,
since, for any context $C \in \M_1 \sqcup \M_2$, writing $i$ for the component to which $C$ belongs, we have:
\begin{calculation}
b|_C
\ejust=
b|_{X_i}|_C 
\just\leq{by eq. \eqref{eq:condition-cmarg}: $b|_{X_i} \leq b_{e_i}$}
b_{e_i}|_C
\just\leq{by feasibility of $b_{e_i}$ for $e_i$}
(e_i)_C
\just={definition of $e_1 \choice e_2$}
(e_1 \choice e_2)_C
\end{calculation}%
This in turn implies that
\begin{align*}
\NCF(e_1 \choice e_2) &\geq \weight{b}
\\ &= \min\enset{\weight{b_{e_1}},\weight{b_{e_2}}}
\\ &= \min\enset{\NCF(e_1),\NCF(e_2)} \Mdot
\end{align*}

It thus remains to show that this $b$ can be contructed. This is achieved by Lemma~\ref{lemma:subdistchoice} proved below. 
%
%
%
%
%
%
%
%
%
%
\end{proof}

\begin{lemma}\label{lemma:subdistchoice-sameweight}
Let $b_S$ and $b_T$ be subprobability distributions on sets $S$ and $T$, respectively, with the same weight $\weightz$.
Then there exists a subprobability distribution $b$ on $S \times T$ with weight $\weightz$ whose marginals are the original subdistributions.
\end{lemma}
\begin{proof}
Let $S_+ \subseteq S$ and $T_+ \subseteq T$ stand for the (finite) supports of the distributions $b_S$ and $b_T$, respectively.
Choose (any) total orderings of the sets $S_+$ and $T_+$, \ie:
\[S_+ = \enset{s_1, \ldots, s_n} \Mand T_+ = \enset{t_1, \ldots, t_m}\Mdot\]
Define the functions $\fdec{L_S,R_S}{S_+}{[0,\weightz]}$ as follows: for each $s_i \in S_+$,
\begin{align*}
L_S(s_i) &\defeq \sum_{1 \leq j < i}b_S(s_j) 
\\
R_S(s_i) &\defeq \sum_{1 \leq j \leq i}b_S(s_j)
\Mdot
\end{align*}
Note that $b_S(s) = R_S(s) - L_S(s)$, and moreover that $R_S(s_i) = L_S(s_{i+1})$ for all $i \in \enset{1,\ldots,n-1}$.
$R_T$ and $L_T$ are defined analogously based on the subprobability distribution $b_T$.

The subprobability distribution $b$ on $S \times T$
is given,
for any $s \in S_+$ and $t \in T_+$, as
\[
b(s,t) \defeq \min\enset{R_S(s), R_T(t)} \ominus \max\enset{L_S(s),L_T(t)}
\]
where $r \ominus l \defeq \max\enset{r-l,0}$,
and by $b(s,t) \defeq 0$ whenever $s \in S \setminus S_+$ or $t \in T \setminus T_+$.

We show that $b|_S = b_S$; the proof that $b|_T = b_T$ is analogous. Note that this also means that $\weight{b} = \weight{b_S} = \weightz$.

Each $s \in S_+$ determines a partition of $T_+$ into the following disjoint subsets (the letters stand for `to the left', `overlapping', and `to the right'):
\begin{align*}
\Lset{s}  &\defeq \setdef{t \in T_+}{R_T(t) \leq L_S(s)}
\\
\Oset{s}  &\defeq \setdef{t \in T_+}{L_S(s) < R_T(t) \land L_T(t) < R_S(s)}
\\
\Rset{s}  &\defeq \setdef{t \in T_+}{L_T(t) \geq R_S(s)}
\end{align*}
These sets satisfy the property that
$l < o < r$ for any $t_l \in \Lset{s}$, $t_o \in \Oset{s}$, and $t_r \in \Rset{s}$.
Observe that $b(s,t) \neq 0$ if and only if $t \in \Oset{s}$, and in that case,
$b(s,t) = \min\enset{R_S(s), R_T(t)} - \max\enset{L_S(s),L_T(t)}$.

Since $R_S(s) - L_S(s) = b_S(s) > 0$, one must have $\Oset{s} \neq \emptyset$.
So, let $t_p$ and $t_u$ be, respectively, the first and last elements of $\Oset{s}$
(note that we are not excluding the possibility that $\Oset{s}$ has a single element; in that case we merely have $p = u$).
Observe that 
\[\max\enset{L_S(s),L_T(t_p)} = L_S(s) \Mcomma\]
for otherwise
$R_T(t_{p-1}) = L_T(t_p) > L_S(s)$ meaning that $t_{p-1} \not\in \Lset{s}$,
which would contradict the minimality of $t_p$ in $\Oset{s}$.
On the other hand, for any other element of $\Oset{s}$, \ie for any $t_j$ with $p < j \leq u$, 
\[\max\enset{L_S(s),L_T(t_j)} = L_T(t_j)\Mcomma\]
for otherwise $L_S(s) > L_T(t_j) \geq R_T(t_p)$ and we would have $t_p \not\in \Oset{s}$, a contradiction.

Dually, we have 
\[\min\enset{R_S(s),R_T(t_u)} = R_S(s) \Mcomma\]
and 
\[\min\enset{R_S(s),R_T(t_j)} = R_T(t_j)\] for any $t_j$ with $p \leq j < u$.

Therefore,
\begin{calculation}
b|_S(s)
\ejust=
\sum_{t \in \Oset{s}}b(s,t)
\ejust=
\sum_{p \leq j \leq u} \left(\min\enset{R_S(s), R_T(t_j)} - \max\enset{L_S(s),L_T(t_j)}\right)
\ejust=
\min\enset{R_S(s), R_T(t_u)}
+
\sum_{p \leq j < u} \min\enset{R_S(s), R_T(t_j)}
\ejust~
-
\sum_{p < j \leq u} \max\enset{L_S(s),L_T(t_j)}
-
\max\enset{L_S(s),L_T(t_p)}
\ejust=
R_S(s)
+
\sum_{p \leq j < u}  R_T(t_j)
-
\sum_{p < j \leq u} L_T(t_j)
-
L_S(s)
\ejust=
R_S(s)
+
\sum_{p \leq j < u} L_T(t_{j+1})
-
\sum_{p < j \leq u} L_T(t_j)
-
L_S(s)
\ejust=
R_S(s)
+
\sum_{p < j \leq u} L_T(t_j)
-
\sum_{p < j \leq u} L_T(t_j)
-
L_S(s)
\ejust=
R_S(s) - L_S(S)
\ejust=
b_S(s)
\end{calculation}%
\end{proof}

\begin{lemma}\label{lemma:subdistchoice}
Let $b_S$ and $b_T$ be subprobability distribution on sets $S$ and $T$, respectively.
Then there exists a subprobability distribution $b$ on $S \times T$ with weight $\weight{b} = \min\{\weight{b_S},\weight{b_T}\}$ such that
$b|_S \leq b_S$ and $b|_T \leq b_T$.
\end{lemma}
\begin{proof}
If one of the distributions has zero weight, the result is obvious, so let $\weight{b_S}, \weight{b_T} > 0$.
Without loss of generality, assume $\weight{b_S} \leq \weight{b_T}$ and renormalise
$b_T$ by the shrinking factor $\frac{\weight{b_S}}{\weight{b_T}}$, yielding a subprobability distribution $b'_T$ that is a subdistribution of $b_T$
and has the same weight as $b_S$. The result then follows from Lemma~\ref{lemma:subdistchoice-sameweight}.
\end{proof}

\section{D. Contextual fraction and \ltMBQC\ (Proof of Theorem~\ref{thm:CFMBQC})}

We recall (and rephrase) from \cite{RaussendorfSC} the definitions of measurement-based quantum computation with $\ZZ_2$-linear classical processing (\ltMBQC).
A computation of this kind is performed
by a parity computer, acting as the classical control to choose measurement settings and processing outcomes,
with access to a resource in the form of an empirical model on a multipartite scenario.

Note that in an $(n,2,2)$ Bell scenario---\ie a scenario where $n$ parties can each choose between performing one of two different measurements, which may each yield one of two possible outcomes---the measurement contexts, which comprise a choice of measurement setting for each party, can be represented by a  vector in $2^n$.
Similarly, a joint outcome may also be represented as a vector in $2^n$.
Therefore, such an empirical model $e$ determines a function $2^n \longrightarrow \Distr (2^n)$ that associates to each (measurement) vector $\qq \in 2^n$ the
probability distribution $e_\qq$ on outcome vectors in $2^n$.

An \ltMBQC\ with $m$ bits of input and $l$ bits of output using an $n$-partite resource consists of:
\begin{itemize}
\item a pre-processing $n \times m$ $\ZZ_2$-matrix $\MQ$;
\item a post-processing $l \times n$ $\ZZ_2$-matrix $\MZ$;
\item an $n \times n$ strictly lower triangular $\ZZ_2$-matrix $\MT$ representing the flow;
\item an empirical model on the $(n,2,2)$ Bell scenario.
\end{itemize}
We shall often denote an \ltMBQC\ as a pair $\tuple{K,e}$
where $K$ is the description of the classical processing (the triple of matrices $\tuple{\MQ,\MT,\MZ}$) and $e$ the empirical model used as a resource.

Each execution starts with a vector $\ii \in 2^m$ of inputs and calculates a vector $\oo \in 2^l$ of outputs,
using two intermediate vectors $\qq, \sv \in 2^n$, proceeding according to:
\begin{equation}\label{eq:l2MBQC-exectution}
\qq \defeq \MQ \ii + \MT \sv
\;\;\Mcomma\;\;\;
\oo \defeq \MZ \sv \Mcomma
\end{equation}
and with $\sv$ obtained from $\qq$ by sampling the distribution $e_\qq$ -- \ie by performing the measurements, using the resource.
Note that in actual fact the distribution is not sampled \textit{at once}, since parties may be asked to perform their measurements at different times.
Indeed, the decision of which measurement to perform on a party may depend on the outcome of the measurement already performed by other parties, as is clear from the way the vector $\qq$ is obtained. The reason for $\MT$ being strictly lower triangular is exactly to ensure that there is an order in which the computation may be performed
(the point is that the $j$-th component of the vector $\qq$, $\qq[j]$, \ie the measurement to be performed by the $j$-th party, can only depend on the values of $\sv[i]$ with $i < j$, \ie on the outcomes of the measurements performed at parties with index strictly smaller than $j$).
But the fact that $e$ satisfies the no-signalling condition guarantees that a subset of the measurements of a context may be performed unambiguously without the knowledge of what the full context will be.

Note that this execution is probabilistic. As such, the \ltMBQC\ $\tuple{K,e}$ determines a map $\fdec{\ds{K,e}}{2^m}{\Distr(2^l)}$
associating to each input bit string a distribution on output bit strings: 
given $\ii \in 2^m$ and $\oo \in 2^l$, $\ds{K,e}(\ii)(\oo)$ is the probability of obtaining output bit string $\oo$ when performing the computation $\tuple{K,e}$ with input bit string $\ii$.

The first ingredient necessary to state the result is a measure of success of the computation in implementing a particular objective function.
Consider a function $\fdec{f}{2^m}{2^l}$, which one aims to implement.
The \ltMBQC\ $\tuple{K,e}$ evaluates $f$ with worst-case success probability
\[p_S^{\tuple{K,e},f} \defeq \min_{\ii \in 2^m} \ds{K,e}(\ii)(f \ii) \Mcomma\]
and with average success probability
\begin{equation}\label{eq:def-ps}
\bar{p}_S^{\tuple{K,e},f} \defeq \frac{1}{2^m}\sum_{\ii \in 2^m} \ds{K,e} (\ii)(f \ii) \Mdot
\end{equation}
We shall omit the upper indices if they are clear from the context. Clearly, we always have $p_S \leq \bar{p}_S$.

We now introduce a measure of the \textit{hardness} of the problem one aims to implement.
Since linear functions are the \textit{free} computations in this model, this expresses how much the objective function deviates from being linear.
The \emph{average distance} between two Boolean functions $\fdec{f,g}{2^m}{2^l}$ is given by 
\begin{equation}\label{eq:def-dist}
\tilde{d}(f,g) \defeq \frac{1}{2^m}|\setdef{i\in 2^m}{f(\ii)\neq g(\ii)}| \Mdot
\end{equation}
The average distance of $f$ to the closest linear function is denoted by 
\begin{equation}\label{eq:def-tildenu}
\tilde{\nu}(f) \defeq \min\setdef{\tilde{d}(f,h)}{\text{$\fdec{h}{2^m}{2^l}$ $\ZZ_2$-linear}} \Mdot
\end{equation}

We can now present a proof of the result that relates the hardness of the problem, the probability of success, and the contextual fraction of the resource.
This sharpens ideas implicit in the proof of \cite[Lemma 1]{RaussendorfSC}.

\thmCFMBQCSUPP
\begin{proof}
Write $e = \NCF(e)e^{NC} + \CF(e)e^{SC}$, and consider the \ltMBQC s $\tuple{K,e^{NC}}$ and $\tuple{K,e^{SC}}$ 
that correspond to using the same classical processing as our \ltMBQC\ of interest but using the empirical model $e^{NC}$ (resp. $e^{SC}$) instead of $e$.
Then, we have that, for any input $\ii \in 2^m$,
\begin{equation}\label{eq:decompose-Ke}
\ds{K,e}(\ii) = \NCF(e)\ds{K,e^{NC}}(\ii) + \CF(e)\ds{K,e^{SC}}(\ii) \Mdot
\end{equation}
and so
\begin{calculation}
\bar{p}_S^{\tuple{K,e},f}
\just={definition of $\bar{p}_S$, eq. \eqref{eq:def-ps}} 
\frac{1}{2^m}\sum_{\ii \in 2^m} \ds{K,e} (\ii)(f \ii)
\just={decomposition of the computation $\tuple{K,e}$, eq. \eqref{eq:decompose-Ke}}
\frac{1}{2^m}\sum_{\ii \in 2^m} ( \NCF(e)\ds{K,e^{NC}}(\ii)(f \ii) 
\\ & &
\hspace{3.6em}+\; \CF(e)\ds{K,e^{SC}}(\ii)(f \ii) )
\just={distributivity and definition of $\bar{p}_S$, eq. \eqref{eq:def-ps} (twice)}
\NCF(e)\, \bar{p}_S^{\tuple{K,e^{NC}},f} \;+\; \CF(e)\, \bar{p}_S^{\tuple{K,e^{SC}},f}
\end{calculation}%

We can bound the probability of success by ignoring what happens when the strongly contextual part of the resource is used (\ie assuming that it always suceeds in that case):
\[
\bar{p}_S^{\tuple{K,e},f} \leq \NCF(e)\, \bar{p}_S^{\tuple{K,e^{NC}},f} \;+\; \CF(e) \Mdot
\]
Hence,
\[\bar{p}_F^{\tuple{K,e},f} \quad \geq \quad 
\NCF(e) \, \bar{p}_F^{\tuple{K,e^{NC}},f} \Mdot\]
That is, the computation will fail to compute $f$ at least when the non-contextual part of the resource is used and fails to compute $f$,
thus the overall average probability of failure will be at least the average probability of failure of the non-contextual part of the resource multiplied
by the weight of this part, \ie the non-contextual fraction.

It remains to show that the average probability of failure of a non-contextual empirical model is at least $\nu(f)$.
One can similarly break down $e^{NC}$ into a convex combination of deterministic non-contextual empirical models,
\[e^{NC} = \sum_{g \in O^X} d(g) \delta_g \]
with $d$ a distribution on $O^X$,
where we write $\delta_g$ for the deterministic model determined by the global assignment $g \in O^X$, for which $(\delta_g)_C$ is the delta distribution at $g|_C$.
By a similar derivation as above, we have
\[
\ds{K,e^{NC}}(\ii)
=
\sum_{g \in O^X}d(g)\, \ds{K,\delta_g}(\ii) \Mcomma
\]
and then
\[
\bar{p}_F^{\tuple{K,e},f}
=
\sum_{g \in O^X}d(g) \, \bar{p}_F^{\tuple{K,\delta_g},f} \Mdot
\]
So, it is enough to show that the average probability of failure for any deterministic non-contextual model $\delta_g$ is at least $\nu(f)$.

Note that computation by $\tuple{K,\delta_g}$ is always deterministic, hence there is a function $\fdec{h^{K,g}}{2^m}{2^l}$ such that
\begin{equation}\label{eq:ds-deltag}
\ds{K,\delta_g}(\ii)(\oo) = \delta(h^{K,g}(\ii),\oo) =
\begin{cases}
1 & \text{ if $h^{K,g}(\ii)=\oo$}
\\
0 & \text{ otherwise}
\end{cases}
\end{equation}

Moreover, this function $h^{K,g}$ is linear, as shown by Raussendorf  \cite[Theorem 2]{RaussendorfSC}.
The proof of this fact essentially comes down to the fact that this function is obtained from composition of several linear functions.
The crux of the matter is the transformation from $\qq$ to $\sv$ (a map $2^n \longrightarrow 2^n$), when determined by a global assignment $g$ as is the case here,
can be seen to be built from
$n$ maps of type $2 \longrightarrow 2$, each of which is necessarily $\ZZ_2$-linear, as are all functions of this type.
Consequently, one can write $\sv = \diag(\dd) \qq + \cc$ for some $n$-component vectors $\dd$ and $\cc$.
Combining this with eq. \eqref{eq:l2MBQC-exectution}, one obtains
\[\oo = \MZ\sv = \MZ (\MIn - \diag(\dd) \MT)^{-1}(\diag(\dd) \MQ \ii + \cc) \Mcomma\]
where invertibility of $\MIn - \diag(\dd) \MT$ follows from $\MT$ being strictly lower triangular. This shows that $\oo$ is a given as a linear function of $\ii$ when the map from  $\qq$ to $\sv$ is determined by a global assignment; that is, $h^{K,g}$ is linear.

As a consequence, we have that
\begin{calculation}
\bar{p}_F^{\tuple{K,\delta_g},f}
\ejust=
1-\bar{p}_S^{\tuple{K,\delta_g},f}
\just={definition of $\bar{p}_S$, eq. \eqref{eq:def-ps}}
1-\frac{1}{2^m}\sum_{\ii \in 2^m} \ds{K,\delta_g} (\ii)(f \ii)
\just={expanding $\ds{K,\delta_g}$ by eq. \eqref{eq:ds-deltag}}
1-\frac{1}{2^m}\sum_{\ii \in 2^m} \delta(h^{K,g}(\ii),f \ii)
\ejust=
1-\frac{1}{2^m}\left|\setdef{\ii \in 2^m}{h^{K,g}(\ii) = f \ii}\right|
\ejust=
\frac{1}{2^m}\left|\setdef{\ii \in 2^m}{h^{K,g}(\ii) \neq f \ii}\right|
\just={definition of average distance, eq. \eqref{eq:def-dist}}
\tilde{d}(h^{K,g},f)
\just\geq{definition of $\tilde{\nu}$, eq. \eqref{eq:def-tildenu}, since $h^{K,g}$ is linear}
\tilde{\nu}(f) \Mcomma
\end{calculation}%
which concludes the proof.
\end{proof}

\section{E. Contextual fraction and games}

We expand here on the last paragraph of the main text.
We discuss an interpretation of the result of Theorem~\ref{thm:CFBellineqs}
in light of the logical description of Bell inequalities
from \cite{AbramskyHardy:LogicalBellIneqs}.
This yields a result with a similar flavour to that of Theorem~\ref{thm:CFMBQC}
relating the hardness of a task, its probability of success, and the contextual fraction of the resource used.

A constraint system is a specified by a tuple $\tuple{V,D,\Gamma}$,
where $V$ is a finite set of variables,
$D$ a finite set called the domain,
and $\Gamma$ a finite set of
formulae on the variables in $V$.
We write $V(\phi)$ for the variables that appear in the formula $\phi \in \Gamma$.

We consider the following task.
Given a formula $\phi \in \Gamma$ as input, one must reply with an assignment
$\fdec{s}{V(\phi)}{D}$ of domain values to each variable that appears in the formula.
Thus, a probabilistic strategy is a family $\enset{p_\phi}_{\phi \in \Gamma}$
where $p_\phi$ is a probability distribution over $D^{V(\phi)}$.
A strategy is considered valid if these probabilities satisfy the following compatibility condition:
\[
\Forall{\phi_1, \phi_2 \in \Gamma}
p_{\phi_1}|_{V(\phi_1) \cap V(\phi_2)} = p_{\phi_2}|_{V(\phi_1) \cap V(\phi_2)}
\]
The idea is that, averaging over several runs, one should give consistent answers for the same variable apearing in different formulae.
In specific cases where this game can be interpreted as a multi-player nonlocal game (see below),
this condition corresponds to the imposed requirement of no communication between players.

The goal is to answer with an assignment $\fdec{s}{V(\phi)}{D}$ that satisfies the input formula $\phi$ (written $s \models \phi$) as often as possible, while following a valid strategy. The \emph{average probability of success} of a strategy $\enset{p_\phi}_{\phi \in \Gamma}$ is given by
\[p_S \defeq \frac{1}{|\Gamma|}\sum_{\phi \in \Gamma}\sum_{s\in D^{V(\phi)}, s \models \phi} p_\phi(s)\]
where we are assuming that all input formulae are equally probable.

Since the aim is to jointly satisfy a set of formulae, the notion of $k$-consistency is a measure of the hardness of the task. A set of formulae is said to be $k$-consistent if at most $k$ can be jointly satisfied by an assignment. If our set $\Gamma$ of $n$ formulae is $k$-consistent, the fraction $\frac{(n-k)}{n}$ is a normalised measure of the hardness of the task.

Note that all this is simply an alternative way of formulating measurement scenarios and empirical models.
We can consider a scenario $\XMO$ given by $X = V$, $O = D$, and with the contexts in $\M$ being the maximal sets of the form $V(\phi)$ with $\phi \in \Gamma$. Then, a valid strategy is simply a (no-signalling) empirical model for this scenario. As such, we can speak of its contextual fraction.

The particular case of Bell scenarios corresponds to a multi-player game where each player is responsible to answer for certain variables and where at most one variable from each player may appear in each formula $\phi \in \Gamma$. 
This is a nonlocal game as considered e.g. in \cite{CleveMittal2014, CleveLiuSlofstra17}, and there the validity condition on the strategy can be motivated as imposing no communication between players.
In fact, there is a way of converting a general constraint satisfaction task into
a two-player nonlocal game in such a way that the imposition of no-communication captures the strategy validity requirement. This transformation yields
an interesting correspondence even at the level of quantum realisability \cite{CleveMittal2014, CleveLiuSlofstra17, MancinskaRoberson2016, AbramskyEtAl2017:QuantumMonad}, where it can be seen as establishing a tight connection between strong nonlocality and state-independent contextuality \cite{AbramskyEtAl2017:QuantumMonad}.
However, we shall not expand on this point here.

Combining our Theorem~\ref{thm:CFBellineqs} with the logical description of Bell inequalities from \cite{AbramskyHardy:LogicalBellIneqs}, we obtain the following result.
\begin{theorem}
Let $\tuple{V,D,\Gamma}$ be a constraint system where $\Gamma$ is $k$-consistent,
and consider a valid strategy $\{p_\phi\}_{\phi \in \Gamma}$ with average probability of success $p_S$, and corresponding average probability of failure $p_F \defeq 1 - p_S$. Then,
\[p_F \geq \NCF(p)\frac{n-k}{k} \Mdot\]
\end{theorem}
\begin{proof}
From \cite[Proposition~V.1]{AbramskyHardy:LogicalBellIneqs},
the $k$-consistent set of formulae $\Gamma$ gives the Bell inequality
\begin{equation*}\label{eq:Bellineq-Kconsistent}
\sum_{\phi \in \Gamma} \; \sum_{\substack{s \in D^{V(\phi)} \\ s\models \phi}} p_\phi(s) \leq k
\end{equation*}
with algebraic bound equal to $n = |\Gamma|$.
Note that the left-hand side is equal to $n \, p_S$.

By Theorem~\ref{thm:CFBellineqs},
we must have that the normalised violation of this Bell inequality by (the empirical model corresponding to) the strategy $\{p_\phi\}_{\phi\in\Gamma}$ is at most $\CF(p)$:
\[
\frac{\max\{0, n \, p_S - k\}}{n - k}
\leq
\CF(p)
\Mdot
\]
By substituting $(1-p_F)$ for $p_S$ and $(1-\NCF(p))$ for $\CF(p)$ and simplifying,
this easily seen to be equivalent to the desired relation.
\end{proof}

\fi
\end{document}